\documentclass[10pt,twoside,english]{article}
\usepackage{ae,aecompl}

\usepackage[T1]{fontenc}
\usepackage[latin9]{inputenc}
\usepackage{geometry}
\usepackage{float}
\usepackage{amsmath}
\usepackage{amsthm}
\usepackage{amssymb}
\usepackage{graphicx}
\usepackage{setspace}
\usepackage[authoryear]{natbib}
\makeatletter
\numberwithin{equation}{section}
  \theoremstyle{remark}
  \newtheorem*{acknowledgement*}{\protect\acknowledgementname}
\newcommand{\lyxaddress}[1]{
\par {\raggedright #1
\vspace{1.4em}
\noindent\par}
}
 \theoremstyle{definition}
 \newtheorem*{defn*}{\protect\definitionname}
\theoremstyle{plain}
\newtheorem{thm}{\protect\theoremname}
  \theoremstyle{plain}
  \newtheorem{cor}[thm]{\protect\corollaryname}
  \theoremstyle{plain}
  \newtheorem*{thm*}{\protect\theoremname}

\makeatother

\usepackage{babel}
  \providecommand{\acknowledgementname}{Acknowledgement}
  \providecommand{\corollaryname}{Corollary}
  \providecommand{\definitionname}{Definition}
  \providecommand{\theoremname}{Theorem}
\providecommand{\theoremname}{Theorem}

\begin{document}

\title{A Model of Collaboration Network Formation with Heterogeneous Skills
}

\author{Katharine A. Anderson\\
Carnegie Mellon University, Tepper School of Business\\
andersok@andrew.cmu.edu}
\maketitle
\begin{abstract}
Collaboration networks provide a method for examining the highly heterogeneous
structure of collaborative communities. However, we still have limited
theoretical understanding of how individual heterogeneity relates
to network heterogeneity. The model presented here provides a framework
linking an individual\textquoteright s skill set to her position in
the collaboration network, and the distribution of skills in the population
to the structure of the collaboration network as a whole. This model
suggests that there is a non-trivial relationship between skills and
network position: individuals with a useful combination of skills
will have a disproportionate number of links in the network. Indeed,
in some cases, an individual\textquoteright s degree is non-monotonic
in the number of skills she has\textemdash an individual with very
few skills may outperform an individual with many. Special cases of
the model suggest that the degree distribution of the network will
be skewed, even when the distribution of skills is uniform in the
population. The degree distribution becomes more skewed as problems
become more difficult, leading to a community dominated by a few high-degree
superstars. This has striking implications for labor market outcomes
in industries where production is largely the result of collaborative
effort.\end{abstract}
\begin{acknowledgement*}
Thanks to Scott Page, Ross O'Connell, Lada Adamic, and Robert Willis
for their valuable input. This work was completed with the help of
funding from the NSF and computing resources from the University of
Michigan Center for the Study of Complex Systems. 
\end{acknowledgement*}

\lyxaddress{Corresponding author: Tepper School, 5000 Forbes Ave, Pittsburgh,
PA 15206 \\
(412) 268-6143}

\pagebreak{}

Collaboration is an increasingly important part of economic activity.
As our economy shifts away from manufacturing, towards more knowledge-based
industries, the dominant methods of production have shifted away from
assembly lines towards collaborative team-based production. As a result,
knowledge-based firms look more and more like universities and national
labs, where collaborative problem-solving is a vital part of work.
It is generally accepted that collaboration is vital because it allows
individuals with different skills to pool their efforts and solve
more difficult problems than any of them could alone (\citealp{Hong2001},
\citealp{Polzer2002}, \citealp{Thomas-Hunt2003}, \citealp{Phillips2004}).
As collaboration between heterogeneous workers becomes more central
to economic production, it is increasingly important to consider not
only the overall composition of the workforce, but also the \emph{structure}
of their interactions.

\begin{figure}[h]
\includegraphics[width=0.9\textwidth]{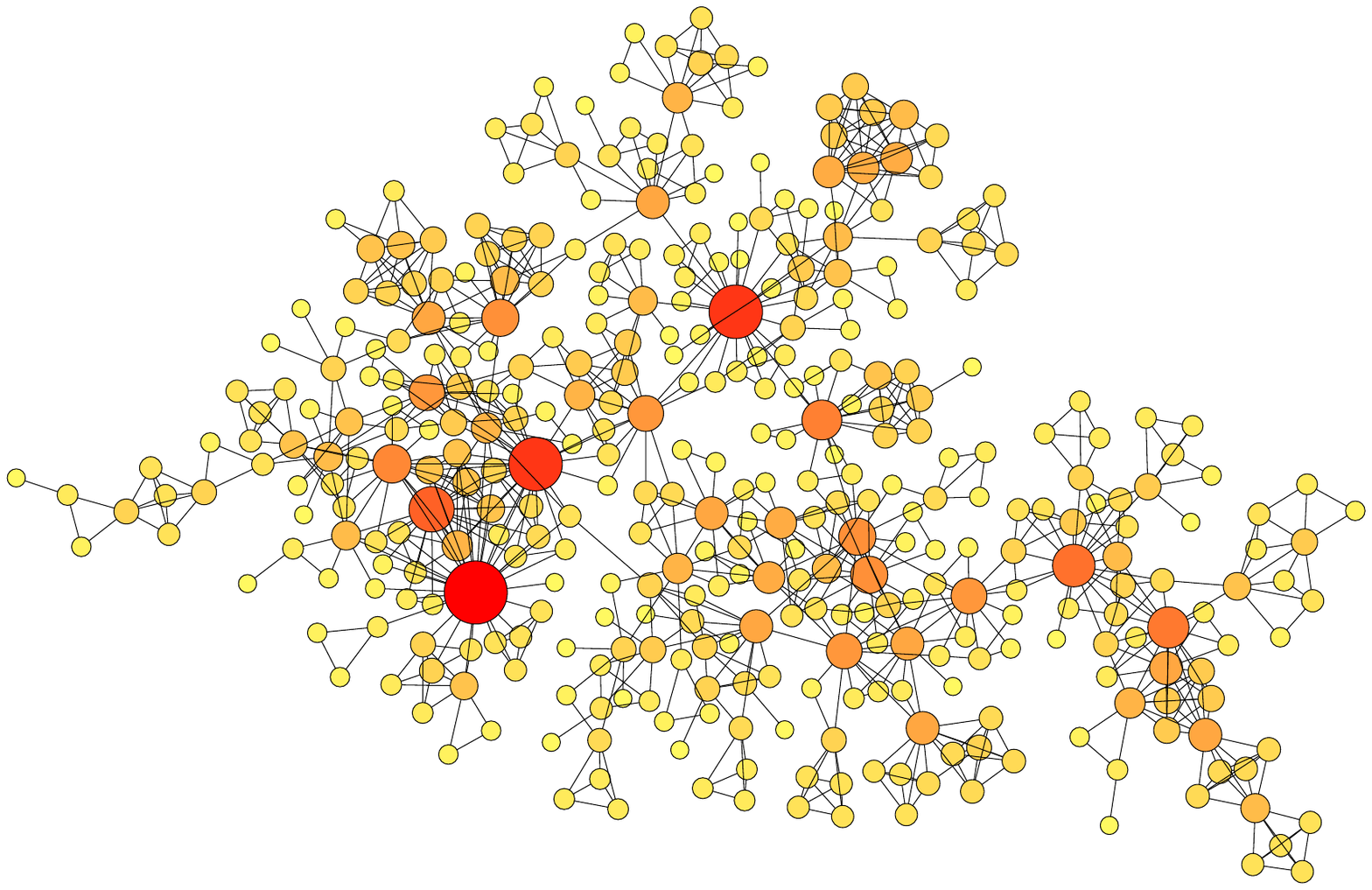}

\caption{\label{fig:econ-coauthorship}A network of coauthorship among network
scientists.}
\end{figure}

Network theory provides a useful perspective when grappling with this
complex set of interactions between collaborators, and a growing literature
has examined the structure of collaboration networks, such as the
network science coauthorship network shown in Figure\ref{fig:econ-coauthorship}
(data courtesy of \citealp{Newman2006}). One of the most striking
features of collaboration networks are their extreme heterogeneity.
Some individuals in this network have many links, while most have
very few. This skewed degree distribution can be found in most collaboration
networks (\citealp{Powell1996}, \citealp{Barabasi1999}, \citealp{Newman2001},
\citealp{Gleiser2003}, \citealp{Moody2004}, \citealp{Uzzi2005},
\citealp{Acedo2006b}, \citealp{Goyal2006}, \citealp{Iyer2006}),
and suggests a similar disparity in both community engagement and
productivity.

There is ample evidence that this network heterogeneity is important
in the function of a collaborative community. On an individual level,
a person\textquoteright s position in the network reflects her role
in the collaborative community, and has an effect on her behavior
and outcomes. Network position is correlated with influence (\citealp{Menzel1955},
\citealp{DeMarzo2003}, \citealp{Golub2010}, \citealp{Banerjee2011}),
access to knowledge, tools, and information (see \citealp{Jackson2008}
for a summary of this literature), and the probability of finding
a job (\citealp{Calvo-Armengol2004} and \citealp{Calvo-Armengol2004a}),
not to mention that individuals with more collaborators tend to be
more productive. On a more global level, the overall structure of
the collaboration network reflects the nature of the problem solving
community\textendash for example, a workplace where a single individual
drives most of the collaborative interaction will function differently
from one where no individual is dominant. In particular, the structure
of the collaboration network affects the flow of information and ideas,
the speed of problem solving efforts, and the spread of effective
technologies (see \citealp{Jackson2007a}, \citealp{Jackson2007}
for examples).

This suggests two interesting classes of questions relating to collaboration
networks: 
\begin{enumerate}
\item What characteristics of an individual affect her position on the collaboration
network? 
\item What aspects of the collaborative process (eg: institutions, characteristics
of problems and problem solving communities) affect the overall structure
of the collaboration network?
\end{enumerate}
A large empirical literature on heterophily in collaboration networks
suggests that skill complementarity is an important driving factor
behind decisions over collaborators (\citet{Moody2004}\citet{Rivera:2010aa}).
Therefore, it seems obvious that an individual\textquoteright s skill
set will govern their role in the collaboration network \textemdash that
individuals occupy different places in the network, not solely due
to luck, but rather due to their different combinations of skills
and abilities\textemdash their human capital. However, despite the
advances made in understanding network formation (described in detail
below) we still have little theoretical understanding of how this
heterogeneity in human capital maps onto observed network heterogeneity.
In this paper, I fill that gap with a model that links an individual\textquoteright s
skill set to her position on the collaboration network. This allows
me to also link the global structure of the collaboration network
to the overall distribution of skills in the population.

I start with a population of individuals, each of whom is endowed
with a set of discrete skills, drawn from a larger pool of skills
that are relevant for solving problems.\footnote{This idea of giving individuals multiple skills is rooted in an older
labor literature, starting with \citealp{Roy1951}, which gives individuals
an ability level in each of two different occupations (for example,
\citealp{Roy1951} gives an individual a skill level as a ``hunter''
and a ``fisherman''). There has been recent work that explicitly
starts to migrate towards a model where individuals can utilize both
skills for tasks (rather than simply using one or the other). Examples
include \citealp{Lazear2004,Lazear2005} and \citealp{Astebro2011}.
However, all of these works confine their analysis two skills. Here,
I consider skill sets of arbitrary size. In a later section, I show
that this distinction is significant. } That set of skills represents the human capital of the problem solver\textendash the
tools, techniques and knowledge that are useful inputs to knowledge-based
production. When problems are difficult, the individuals collaborate
with others who have complementary skills. The result is a collaboration
network. The structure of this network and an individual\textquoteright s
position in it will necessarily differ depending on the distribution
of skills in the population. For example, if individuals in a community
have only one skill apiece, then individuals with rare skills will
have more collaborators. But when individuals potentially have multiple
skills, the picture is not so straightforward. Through several examples,
I will show that an individual with a useful combination of skills
may have more collaborators than an individual with rare skills. Moreover,
degree is sometimes non-monotonic in the size of one\textquoteright s
skill set: individuals with more skills will sometimes have fewer
collaborators than those with many skills. These examples illustrate
the value of a model with a flexible notion of skill heterogeneity. 

In the general model, where the population of problem solvers has
an arbitrary distribution of skills, I show that an individual\textquoteright s
degree on the collaboration network is a supermodular function of
her set of skills. I then use a set of examples to illustrate how
a more complex model of human capital allows for a more realistic
picture of the relationship between skills and collaborative production.
In particular, when individuals are allowed to have multiple skills,
it is possible for the relationship between skills and collaboration
to be not only nonlinear, but \emph{non-monotonic}. 

I then use a special case\textemdash called the Bernoulli Skills Model\textemdash to
show how the process of collaboration can exaggerate small differences
in skills to generate surprisingly large gaps in observed degree.
Small initial differences in individual skill sets create large differences
in the distribution of links, meaning that the distribution of links
in the collaboration network may be skewed, even when the distribution
of skills is not. Moreover, as problems become more difficult, the
degree distribution of the network becomes more skewed, and the network
becomes dominated by a few, high-degree superstars. Using a second
special case, called the Ladder Model, I observe another pattern\textendash as
skills become increasingly hierarchical, the network becomes increasingly
skewed, and superstars emerge. 

There is already an extensive literature on the formation of collaboration
networks, which I have outlined in more detail below. This model provides
insights beyond those offered by these models. On an individual level,
this model provides a useful connection between agent heterogeneity
and outcomes. Preferential attachment provides a dynamic in which
individuals who are already successful might become more successful.
The model presented here connects link formation to skill complementarity,
and thus provides insights into what happens before the preferential
attachment dynamic takes hold. This is pertinent in collaborative
communities, because people enter the system all the time. The question
of which of those individuals will eventually become a star has not
been completely explored in a theoretical framework. On a more global
level, this paper provides several predictions about how collaboration
networks will differ in different populations. The model predicts
that degree distributions will become more skewed as problems become
more difficult. There is some indication that this is true in academic
fields. In particular, data on economists suggests that the degree
distribution of the network has become more unequal over time. If
we believe that research problems have also become more complex over
that time period, then this theory suggests a possible explanation
for that change. By linking the distribution of skills in the population
to the structure of the collaboration network, this model provides
insights into what generates the differences between communities.

The framework and results I present here also have implications for
our understanding of heterogeneity in labor markets. As firms in knowledge-based
industries shift towards team-based production, labor productivity
becomes increasingly tied to collaborative effort. I show that when
workers are allowed to have multiple skills, the resulting distribution
of productivity is much different than when we assume that individuals
have a \textquotedblleft type\textquotedblright{} (e.g.: hunter or
fisherman). Moreover, because skills are synergistic and work in combination,
the relationship between an individual\textquoteright s skill set
and her wages cannot be captured by a linear pricing schedule, where
each skill is valued in isolation. This means that treating skills
in combination may explain more wage variation than more traditional
economic understanding of labor heterogeneity. I include a brief discussion
of these issues below. 

It is also worth noting the broader contribution of this model\textemdash it
provides a general framework for understanding how individual heterogeneity
affects both network structure and collaborative production. Particular
specifications of the fundamentals of the model\textemdash the distribution
of skills, the distribution of problems, the functional form of the
production and payoff functions\textemdash will produce different
outcomes. Thus, the framework introduced here has the potential to
spawn additional work, exploring those relationships. I conclude with
a discussion of some of those possible extensions.

\section{Literature}

This paper draws upon two distinct literatures. The first is the literature
linking skills to productivity and labor market outcomes. This literature
treats skills in one of several canonical ways. One branch of the
literature treats skills in terms of occupation: an individual is
categorized according to what job they perform, and is allowed only
one category. For example, an employee might be a manager or a line
worker, but not a combination of both. This model works well in the
context of manufacturing, where each individual has a well-defined
job, but works less well in the context of knowledge-based production.
Another branch of the literature treats skill level as a one-dimensional
quantity, ultimately placing individuals in either high or low skilled
baskets. This includes the wealth of signaling literature, starting
with \citealp{spence1973job}. This method for modeling skills is
convenient, because it allows a partial ordering over individuals.
However, it seems a better model for some skills (eg: piece work)
than others (eg: problem solving). A third literature takes a kind
of hybrid approach, allowing individuals to have an ability level
in two different skills. This includes the vast literature stemming
from \citealp{Roy1951}. Roy imposes a condition that individuals
use only one of their skills at a time. More recent work has allowed
individuals to use both skills at once. Examples include \citealp{Lazear2004,Lazear2005}
and \citealp{Astebro2011} 

This paper is closest to the last category of models, but represents
a significant generalization from this approach. In particular, individuals
are not confined to two skills, and may have an arbitrarily complex
basket of skills. This model of skill heterogeneity is appealing,
because it better represents the rich, multi-faceted skill sets of
workers in knowledge-based industries. Moreover, as I will show in
Section \ref{sec: Value of Model with Multiple Skills}, allowing
individuals to have three or more skills has a significant impact
on the distribution of productivity in the collaborative community,
making it a crucial extension of previous labor models. Moreover,
by altering the distribution of skills in the population, this model
encompasses all three canonical models of labor heterogeneity, while
also allowing for other, more complex types of heterogeneity.\footnote{It is simple to see how this model subsumes the model with two discrete
skills. However, it is flexible enough to subsume the other two models
as well. In particular, by allowing for skills to build off of one
another (as in Section \ref{sec:Skill-Ladders}), then it is possible
to move away from a binary skill set to one where individuals have
an ability ordering in one or more skill areas.}

This more multi-dimensional model of skills also interfaces well with
the empirical literature on collaborative linking decisions. When
problems are hard, problem-solvers tend to seek out those who have
skills different than their own (\citealp{Rivera:2010aa}). \citealp{Moody2004}
suggests that individuals facing difficult problems seek out those
with complementary skills because it is easier to work with other
researchers than it would be to obtain a new set of skills themselves.
Here, I take this need for skill complementarity in collaboration
as given. 

This paper also contributes to the literature on social network formation.
This literature is vast, and growing rapidly. One branch of the literature
focuses on the role of individual decisions in network formation.
For example, in the Connections Model (\citealp{Jackson1996}) players
gain a benefit for both direct and indirect links and pay a cost for
each direct link made. In the Coauthor Model, \citealp{Jackson1996},
players must allocate their effort across multiple projects where
the payoff from a paper is inversely related to the number of links
the two coauthors have. In \citealp{Goyal2001}, firms choose a set
of links and an effort level to put into research and development,
given that such links general both perfect and imperfect spillovers.
The resulting networks are useful because individuals in decision-based
models respond to incentives, allowing us to see how changes in the
community change the fundamentals of network structure, such as density. 

Another branch of the literature allows links to form via a stochastic
process. In the preferential attachment model, entering nodes link
to existing nodes with a probability proportional to their current
degree. Nodes who gain an early benefit from extra links tend to amass
even more links over time. This model and its derivatives provide
a glimpse into the mechanisms behind the observed skew in degree distribution
(see, among others, \citealp{Barabasi1999}, \citealp{Jackson2007}
and \citealp{Ramasco2007}). The incumbency model of network growth
presented in \citealp{Guimera2005} captures a different dynamic of
collaborative network formation--the role of incumbency. Individuals
create new links via serial team formation. In every round, they form
teams taking into account both whether a potential team member is
new to the community, and any previous working relationships. 

There has also been progress made in linking network heterogeneity
to other kinds of individual heterogeneity: \citealp{Jackson2005}
considers heterogeneity in the costs and benefits of link formation,
\citealp{Galeotti2006353} considers heterogeneity in link costs in
Nash Networks, and \citealp{Carayol2009414} consider the effects
of geographic inhomogeneity. However, none of these previous works
consider a type of heterogeneity crucial to an economic understanding
of collaborative production: labor heterogeneity. 

The contributions of this paper to the literature on network formation
take place on two different levels. On a global level, this model
links differences in network structure to differences in the underlying
population. This provides new insights into why different collaboration
networks have different structures. On the individual level, allowing
individuals to have heterogeneous skill sets allows us to explore
the determinants of network position, and in particular, which individuals
become stars. 

In this paper, I combines these two literatures to create a powerful,
yet flexible theoretical framework for understanding how labor heterogeneity
shapes network formation. The field of labor economics provides insight
into skill heterogeneity, giving us perspective on the role of skill
skill sets in determining an individual's role in a collaboration
networ, and conversely, considering collaboration networks can give
greater understanding of how those skill sets are translated into
labor market outcomes in the case of collaborative, team-based production.

\section{\label{sec:Model}A General Model of Skills, Problem Solving, and
Collaboration Networks}

\subsection{Inputs: Problem Solving Population and Problems}

Let $I=\left\{ 1,2,...N\right\} $ be the set of problem solvers (individuals). 

Let $S=\left\{ a_{1}...a_{M}\right\} $ denote the (finite) set of
all skills. 

A individual $i$'s \emph{skill set }is the subset of those skills
she possesses, $A_{i}\subseteq S$. The frequency of skill set $A$
in the population is given by $\Psi\left(A\right)$, a probability
measure with support $\Sigma\left(\Psi\right)\subseteq2^{S}$--that
is, $\Psi\left(A\right)$ is the fraction of the individuals in $I$
who have the skill set $A\subseteq S$.\footnote{Formally, $\Psi$ is a \emph{frequency }distribution, rather than
a probability distribution--that is, $\Psi$ is a \emph{realized}
distribution of skill sets across the population. The distinction
between frequency and probability distributions disappears when $N$
is large--using a frequency distribution allows me to also make statements
about small $N$ as well.}

Each individual is endowed with a copy of a problem requiring a subset
of the skills available in the population, $\omega\subseteq S$.\footnote{For simplicity, I assume that all individuals face the same problem.
The results are identical if each individual faces a different problem
drawn from a known distribution of problems, $\Omega$, with support
$2^{S}$, and collaborators are chosen \emph{ex ante}. If they choose
their collaborators \emph{ex post}, then the results are similar,
if notationally more complex. } 

Thus, the inputs to the model are a set of skills used to solve problems
($S$) and a population of problem solvers $\left(\Psi\right)$.

\subsection{Collaboration and Problem Solving}

A \emph{collaboration} is a subset of the problem solvers, $C\subset I$.
A collaboration can \emph{solve} a problem if together they possess
all of the required skills--that is, if the problem is solved if $\omega\subseteq\bigcup_{j\in C_{i}}A_{j}$
.

The problem yields a payoff of 1 if solved. If a problem solver can
solve her problem alone (that is, if $\omega\subseteq A_{i}$) then
she keeps the entire payoff. If she solves it with the help of others,
she splits the payoff evenly with her collaborators, giving each a
share of $\frac{1}{\left|C_{i}\right|}$.\footnote{This particular distribution scheme has a number of points to recommend
it: it does not require that agents have control over the way payoffs
are split (advantageous when payoffs are non-monetary), and it gives
each individual her Shapely value for the coalition. Many alternative
splitting schemes will be behaviorally identical (see below).}Since others in the community face their own problems, a problem solver
may be asked to help with other problems as well. Thus, individual
$i$'s payoff is the sum the payoff she gets from solving her own
problem, plus the payoffs she gets from collaborating with others
on their problems: 
\[
u_{i}=\frac{1}{\left|C_{i}\right|}+\sum_{j\ne i\,st\,i\in C_{j}}\frac{1}{\left|C_{j}\right|}
\]
\footnote{It is important to highlight that this payoff is not equivalent to
a market wage. A model producing equilibrium wages would require considerable
additional machinery, including time constraints and search costs,
and is therefore beyond the scope of the current work.} An individual chooses her set of collaborators $\left(C_{i}\right)$
from the pool of all possible collaborators $\left(I\right)$\footnote{Note that this model assumes that an individual looking for collaborators
knows the skill sets of all her colleagues, and thus does not address
the issue of search. This assumption is not problematic when collaborative
communities are small or tight-knit (eg: within firm or in research
subfields). However, in other contexts--particularly ones where particular
skill combinations are very rare--this assumption of global information
may start to become unrealistic. A model incorporating local search
may alter the results presented here, and would be a fertile area
for study.} to maximize her utility. Note that her payoff to solving her own
problem is always positive, and thus it is always incentive compatible
for her to solve the problem. Since each individual controls only
her own collaborative decisions, the utility-maximizer chooses $C_{i}$
to minimize the number of collaborators she must work with on her
own problem--in other words, she chooses a minimal subcover of the
set of skills she lacks--$A_{i}^{c}=\omega_{i}\backslash A_{i}$.
Let $\mathbb{C}_{i}$ denote the set of all minimal subcovers of $A_{i}^{c}$.
If there exist multiple minimal subcovers (ie: if $\left|\mathbb{C}_{i}\right|>1$)
then I assume that the individual chooses $C_{i}^{*}\in\mathbb{C}_{i}$
at random.\footnote{Since the individual indifferent between minimal subcovers, this choice
at random follows convention. The results that follow are not sensitive
to this assumption. In particular, if an individual is biased towards
choosing collaborators who are not as busy (eg: those with fewer links)
then the results that follow hold precisely, rather than on average.
If an individual is biased towards choosing collaborators who are
popular (eg: have many links) than the results that follow are qualitatively
similar, and if anything are more striking. }

It is worth noting that while I have chosen to model the motivate
the decision to minimize the set of collaborators via an equal split
of payoffs, the results that follow are unchanged under any distribution
scheme which induces problem solvers to minimize the number of collaborators
they use. So, for example, if collaborating with additional individuals
is costly (eg: in terms of time or communication), then the results
that follow will still hold with no modification. The results will
not hold under any payoff structure where it is not individually rational
to minimize the set of collaborators used to solve a problem.

\subsection{Complementary Skills Networks}

For a given a set of collaborations, $C=\left\{ C_{1}...C_{N}\right\} $,
the collaboration network is represented by an adjacency matrix, $g\left(C\right)$,
where $g_{ij}\left(C\right)=1$ if $j\in C_{i}$. Note that the network
is directed--since $j\in C_{i}$ does not necessarily imply $i\in C_{j}$,
it may be that $g_{ij}\left(C\right)\ne g_{ji}\left(C\right)$. However,
the links are mutual, in the sense that $j$ will never want to terminate
a link (see Section \ref{sec:Stability-and-Efficiency} for further
discussion). When all collaborators are chosen optimally (that is,
when $C_{i}\in\mathbb{C}_{i}\,\forall i$), I will call the result
a \emph{complementary skills network }(or \emph{complementarity network}
for short). 
\begin{defn*}
A network, $g\left(C\right)$, is a \emph{complementary skills network
(complementarity network) }if each individual in the network chooses
a minimal set of collaborators required to solve her problem--eg:
if $C_{i}\in\mathbb{C}_{i}\,\forall i$.
\end{defn*}
Since the set of minimal subcovers each individual ($\mathbb{C}_{i}$),
depends on the distribution of skills in the population, I use $\Gamma\left(\Psi\right)$
to denote the set of complementary skills networks for a particular
distribution of skills, $\Psi$. In other words, $\Gamma\left(\Psi\right)$
is a set of networks over which all individuals are maximizing utility.
In all of the measures that follow, I will average over all networks
in $\Gamma\left(\Psi\right)$.\footnote{Note that this averaging has the effect of dividing the space of all
possible networks into equivalent classes of networks that are functionally
equivalent from the perspective on individuals in the community. This
is a useful coarse-graining technique that may have applications outside
of the current model.}

\subsection{Network Measures}

Recall that I will be looking at two different questions: first, how
an individual's position on the network depends on her set of skills,
and second how the overall structure of the network depends on the
distribution of skills in the population. For both of these questions,
I will focus on degree in the network. While degree is clearly not
the only network measure defining an individual's position in the
collaborative community, an exploration of other measures is outside
the scope of the current work. 

First, I will look at how an individual's average in-degree on the
networks in $\Gamma\left(\Psi\right)$ depends on her skill set. I
will denote individual $i$'s in-degree on a particular network as
$d_{i}$. Her average in-degree will be $E\left[d_{i}\right]$, where
the expectation is taken over all networks in $\Gamma\left(\Psi\right)$--that
is, all networks that the individual is indifferent between. In Section
\ref{sec:Skills and Degree}, I will consider the mapping between
an individual's skill set an her expected in-degree in a particular
problem-solving population, $E\left[d_{i}\right]=f\left(A_{i},\Psi\right)$. 

I will also consider how the overall distribution of in-degree depends
on the distribution of skills in the population. Let $\Delta$ denote
the distribution of expected in-degree. That is, $\Delta\left(d\right)$
is the fraction of the individuals who have expected in-degree $d$,
where the expectation is taken over all $g\left(C\right)\in\Gamma\left(\Psi\right)$.\footnote{Alternatively, we might plot the distribution of degree across all
networks $g\in\Gamma\left(\Psi\right)$. That is, we could set $\Delta\left(d\right)=\sum_{g\in\Gamma\left(\Psi\right)}\delta_{g}\left(d\right)$
where $\delta_{g}\left(d\right)$ is the fraction of nodes in network
$g$ with degree $d$. This choice does not affect the results.}

Before continuing, a brief word about network notation is in order.
First, note that for ease of reading I will usually drop the argument
of $g\left(C\right)$. I will denote a link from $i$ to $j$ by $ij$.
Using a slight abuse of notation, I will use $g$ to refer to both
the adjacency matrix (as above) \emph{and} the set of links in the
network--that is, $ij\in g$ if $i$ is connected to $j$ in the network
$g$. In a similar abuse of notation, I will use $g-ij$ to represent
the network that results when the link $ij$ is removed from an existing
network, $g$, and $g+ij$ to represent the network that results when
the link $ij$ is added to the existing network, $g$. Finally, for
clarity in the exposition, I will refer to in-degree simply as ``degree''.
This last point should cause no significant confusion--I use in-degree
because it has a clear, empirical interpretation, but the results
qualitatively similar if we consider an individual's degree in the
directed network (the sum of his in-degree and out-degree), or use
the degree of the individual in a network where directed links are
projected into undirected links.

\subsection{Example}

An example will help clarify the structure of this model. Suppose
all of the individuals in the population face the same problem requiring
three skills: $\omega=\left\{ a,b,c\right\} $ $\forall i$ --for
example, developing a web application might require programming skills,
user interface design skills, and marketing skills. Suppose the distribution
of skills is such that everyone in the population has at least one
skill, but no one has all of the skills required. In other words,
every individual in the population has something to add, but none
of them can solve the problem on their own. Suppose further that each
skill combination is equally likely--in other words, $\Psi\left(A\right)=\frac{1}{6}$
for all $A\in\left\{ \left\{ a\right\} ,\left\{ b\right\} ,\left\{ c\right\} ,\left\{ ab\right\} ,\left\{ ac\right\} ,\left\{ bc\right\} \right\} $
and $\Psi\left(A\right)=0$ otherwise (see Figure \ref{fig:PopulationExample}).

\begin{figure}[h]
\includegraphics[clip,width=0.9\textwidth]{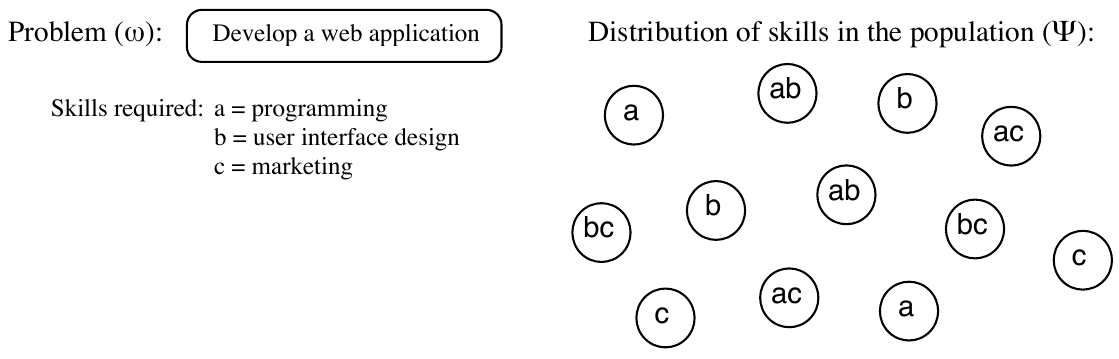}

\caption{\label{fig:PopulationExample}An example of a problem and population
of problem solvers}
\end{figure}
In this particular population of problem solvers, each individual
needs exactly one collaborator to solve the problem--that is, someone
with skill set $\left\{ a\right\} $ must link to someone with the
skill set $\left\{ b,c\right\} $, someone with skill set $\left\{ b,c\right\} $
may choose from those with skill sets $\left\{ a\right\} ,\left\{ a,b\right\} $,
and $\left\{ a,c\right\} $, and so on. The problem solver is indifferent
between any two individuals who have the skills that she needs. By
linking those who collaborate on problems, we obtain a collaboration
network. There will be one such network for every set of optimal collaborations.
Figure \ref{fig:TwoNetworkExamples} shows two collaboration networks
for this population of problem solvers and two optimal choices of
collaborators. The set of all such networks for a given population
of problem solvers is denoted by $\Gamma\left(\Psi\right)$.

\begin{figure}[h]
\includegraphics[clip,width=0.9\columnwidth]{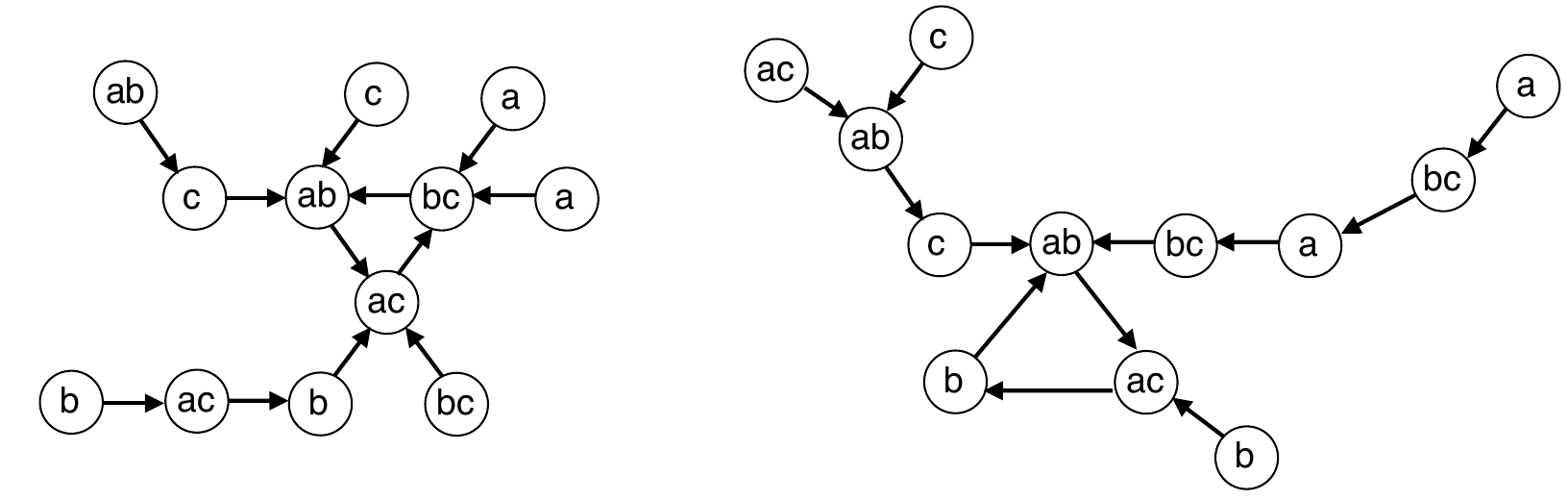}

\caption{\label{fig:TwoNetworkExamples}An example of two networks formed from
population/problem in Figure \ref{fig:PopulationExample}}
\end{figure}

\subsection{Discussion}

This model produces outcomes consistent with several empirical facts
about collaboration and problem solving. First, the model predicts
that collaboration will be more important when the average individual
has a thinner slice of the total skills required to solve a problem--in
other words, as problems become more difficult, collaboration networks
will become more densely connected. This prediction is born out in
the data--collaborative work has become increasingly common in a variety
of academic fields (see \citealp{BARABASI2002} (physics), \citealp{Grossman1995}
(mathematics), \citealp{Moody2004} (sociology), \citealp{Acedo2006b}
(management science), \citealp{Laband2000} and \citealp{Goyal2006}
(economics)). Moreover, the literature supports a connection between
this increased reliance on collaboration and the difficulty of problems
faced, specifically the increasing complexity of required methodologies
(\citealp{Laband2000}and \citealp{Moody2004}). This highlights one
of the advantages of using this model of collaboration network formation:
it allows for a direct connection between problem difficulty and collaborative
effort.

\subsubsection{\label{sec:Stability-and-Efficiency}A Note on Stability and Efficiency
of the Complementarity Network}

Before considering any specific questions about the complementarity
network, it is worth considering it's stability and efficiency. \citealp{Jackson1996}
introduce an equilibrium concept of network stability, called\emph{
pairwise stability.} Briefly, a network is pairwise stable if no individual
would prefer to terminate an existing link, and if no pair of individuals
would prefer to add a link (see Appendix A for a more formal definition
in the case of a directed network). Pairwise stability implies that
links are mutual, because both individuals involved agree to maintain
the link. Theorem \ref{thm:Pairwise stability} states that any complementary
skills network is pairwise stable, implying that all links in the
network are mutual. Because the equilibrium links are mutually beneficial,
one could functionally think of a complementarity network as either
directed (because that is how it is constructed) or undirected (because
that it is how it functions). For ease of reading, I will usually
omit the directional arrows from networks pictured in this paper.

Moreover, any complementary skills network is strongly efficient--that
is, the population extracts the maximum possible value from the network.
This result--that equilibrium networks are efficient--contrasts with
two other models of social network formation--\citealp{Jackson1996}
and \citealp{Goyal2001}--in which pairwise stable networks tend to
have more links than is efficient.
\begin{thm}
\label{thm:Pairwise stability}

Any complementary skills network, $g\in\Gamma\left(\Psi\right)$,
is pairwise stable and strongly efficient. \end{thm}
\begin{proof}
See Appendix A
\end{proof}

\section{\label{sec:Skills and Degree}Skills and Degree in the Collaboration
Network}

In this section, I consider the relationship between an individual's
skill set, and her degree in the collaboration network.

\subsection{Skills and Degree: An Example}

Before presenting general results, it is useful to see an example.
Consider the example from the previous section: each individual faces
a problem requiring three skills, $S=\left\{ a,b,c\right\} $, and
they each have one or two of those skills, so that $\Psi$ has support
$\left\{ \left\{ a\right\} ,\left\{ b\right\} ,\left\{ c\right\} ,\left\{ ab\right\} ,\left\{ ac\right\} ,\left\{ bc\right\} \right\} $. 

An individual's in-degree on the network will depend on the number
of people who need her skills (the demand for a subset of her skills)
and the number of other people who can provide those skills (the supply
of that subset of skills). For example, consider someone with the
skill set $\left\{ a\right\} $. Her skills are in demand by anyone
who needs skill $a$--in this example, anyone with skill set $\left\{ b,c\right\} $.
A person with skill set $\left\{ b,c\right\} $ may collaborate with
anyone who has skill $a$, including those with skill sets $\left\{ a\right\} ,\left\{ a,b\right\} ,$
or $\left\{ a,c\right\} $. So in this example, the expected degree
of an individual with skill set $\left\{ a\right\} $ is 
\[
E\left[d\left(\left\{ a\right\} \right)\right]=\frac{\Psi\left(\left\{ b,c\right\} \right)}{\Psi\left(\left\{ a\right\} \right)+\Psi\left(\left\{ a,b\right\} \right)+\Psi\left(\left\{ a,c\right\} \right)}
\]
 Similarly, an individual with the skill set $\left\{ a,b\right\} $
can help anyone who needs skill $a$, anyone who needs skill $b$,
and anyone who needs skills $a$ and $b$, yielding expected degree{\small{}
\begin{eqnarray*}
E\left[d\left(\left\{ a,b\right\} \right)\right] & = & \frac{\Psi\left(\left\{ b,c\right\} \right)}{\Psi\left(\left\{ a\right\} \right)+\Psi\left(\left\{ a,b\right\} \right)+\Psi\left(\left\{ a,c\right\} \right)}+\frac{\Psi\left(\left\{ a,c\right\} \right)}{\Psi\left(\left\{ b\right\} \right)+\Psi\left(\left\{ a,b\right\} \right)+\Psi\left(\left\{ b,c\right\} \right)}+\frac{\Psi\left(\left\{ c\right\} \right)}{\Psi\left(\left\{ b,c\right\} \right)}
\end{eqnarray*}
}Note that an individual with skills $a$ and $b$ will have more
links, on average, than an individual with skill $a$ and an individual
with skill $b$ put together: $E\left[d\left(\left\{ a,b\right\} \right)\right]>E\left[d\left(\left\{ a\right\} \right)\right]+E\left[d\left(\left\{ b\right\} \right)\right]$.
This is because an individual with both skills can help anyone who
needs skill $a$, anyone who needs skill $b$, \emph{and anyone who
needs both. }This means that an individual with both of those skills
will have, on average, more than twice the number of links than a
person with either of those skills in isolation.

\subsection{Skills and Degree: General Results}

A similar type of result holds more generally. Theorem \ref{thm:SuperModularity}
states that an individual's expected degree in a complementary skills
network is a supermodular function of her set of skills. That means
that regardless of the skill set required for the problem or the distribution
of skills in the population, an individual with skill set $A\cup B$
will have at least as many links as individuals with skill sets $A$
and $B$ put together. The intuition for this result is similar to
the above example--the set of all problems that can be solved by someone
with the skill set $A\cup B$ includes those that can be solved by
someone with skill set $A$, and those that can be solved by someone
with skill set $B$, plus those requiring some skills from \emph{both}
sets. 
\begin{thm}
\label{thm:SuperModularity}For any set of skills, $S$, and distribution
of those skills, $\Psi$, an individual's expected degree over the
networks in $\Gamma\left(\Psi\right)$ is a supermodular function
of her set of skills. That is, $Ed\left(A\cup B\right)+Ed\left(A\cap B\right)\ge Ed\left(A\right)+Ed\left(B\right)$.\end{thm}
\begin{proof}
Here, I will prove the result for the case where an individual needs
only one collaborator to solve her problem. For the sake of clarity,
also I consider the case where $A\cap B=\emptyset$. The proof for
the general result is similar, and can be found in Appendix B. Since
$d\left(A\cap B\right)=d\left(\emptyset\right)=0$, we need to show
that $Ed\left(A\cup B\right)\ge Ed\left(A\right)+Ed\left(B\right)$.
Consider $d\left(A\cup B\right)$. The fraction of the population
needing $C\subseteq A\cup B$ is $\delta\left(C\right)=\Psi\left(S\backslash C\right)$.
The fraction who can supply the set $C$ is $\sigma\left(C\right)=\sum_{D\subseteq S\backslash C}\Psi\left(C\cup D\right)$.
Thus 
\[
E\left[d\left(A\cup B\right)\right]=\sum_{C\subseteq A\cup B}\frac{\Psi\left(S\backslash C\right)}{\sum_{D\subseteq S\backslash C}\Psi\left(C\cup D\right)}=\sum_{C\subseteq A\cup B}\frac{\delta\left(C\right)}{\sigma\left(C\right)}
\]

We can partition $C\subseteq A\cup B$ into three categories: 

1) $C\subseteq A$ 

2) $C\subseteq B$ 

3) $C\subseteq A\cup B$, $C\nsubseteq A$ , $C\nsubseteq B$. 

This division gives us the following: 
\begin{eqnarray*}
E\left[d\left(A\cup B\right)\right] & = & \sum_{C\subseteq A}\frac{\delta\left(C\right)}{\sigma\left(C\right)}+\sum_{C\subseteq B}\frac{\delta\left(C\right)}{\sigma\left(C\right)}+\sum_{C\subseteq A\cup B\,and\,C\cap A,\,C\cap B\ne\emptyset}\frac{\delta\left(C\right)}{\sigma\left(C\right)}\\
 & = & E\left[d\left(A\right)\right]+E\left[d\left(B\right)\right]+\phi\\
 & \ge & E\left[d\left(A\right)\right]+E\left[d\left(B\right)\right]
\end{eqnarray*}

\end{proof}
This theorem suggests an immediate corollary. 
\begin{cor}
\label{cor:Monotonicity} Adding skills to an individual's skill set
will never decrease her average degree in a cost minimizing collaboration
network.\end{cor}
\begin{proof}
From Theorem \ref{thm:SuperModularity}, $E\left[d\left(A\cup a\right)\right]+E\left[d\left(A\cap a\right)\right]=E\left[d\left(A\cup a\right)\right]\ge E\left[d\left(A\right)\right]+E\left[d\left(a\right)\right]$,
and so $E\left[d\left(A\cup a\right)\right]-E\left[d\left(A\right)\right]\ge E\left[d\left(a\right)\right]\ge0$.
\end{proof}

\section{\label{sec: Value of Model with Multiple Skills} The Importance
of Skill Combinations and Complex Heterogeneity}

The result presented in Theorem \ref{thm:SuperModularity} has some
implications for who we think is important in a collaborative community.
In particular, it indicates that it is important to consider an individual's
\emph{combination }of skills, rather than looking at each of her skills
in isolation. When actors are able to consider a potential collaborator's
skills in combination, the result is much different than one would
see in a model where skills are evaluated individually.

As an illustration, suppose there are two different problem solving
populations, faced with the same problem, $\omega=\left\{ a,b,c\right\} $.
In population A, every individual has exactly one skill, which we
can think of as her speciality or ``type'': $\Psi\left(a\right)=\frac{1}{3},\,\Psi\left(b\right)=\frac{1}{3},\Psi\left(c\right)=\frac{1}{3}$.
In population B, the skills are distributed independently. This means
that and the probability of having skill set $A$ is $\Psi\left(A\right)=\prod_{i\in A}p_{i}$.
Let $p_{i}=Prob\left(have\,skill\,i\right)=\frac{1}{3}$ for $i=a,b,c$.

These two populations have much in common: in both, the three skills
occur in equal proportion, and both populations average one skill
per person. However, as illustrated in Figure \ref{fig:Two populations Network Structure},
the collaboration network is much different in the two different populations.

\begin{figure}[h]
\includegraphics[width=1\columnwidth]{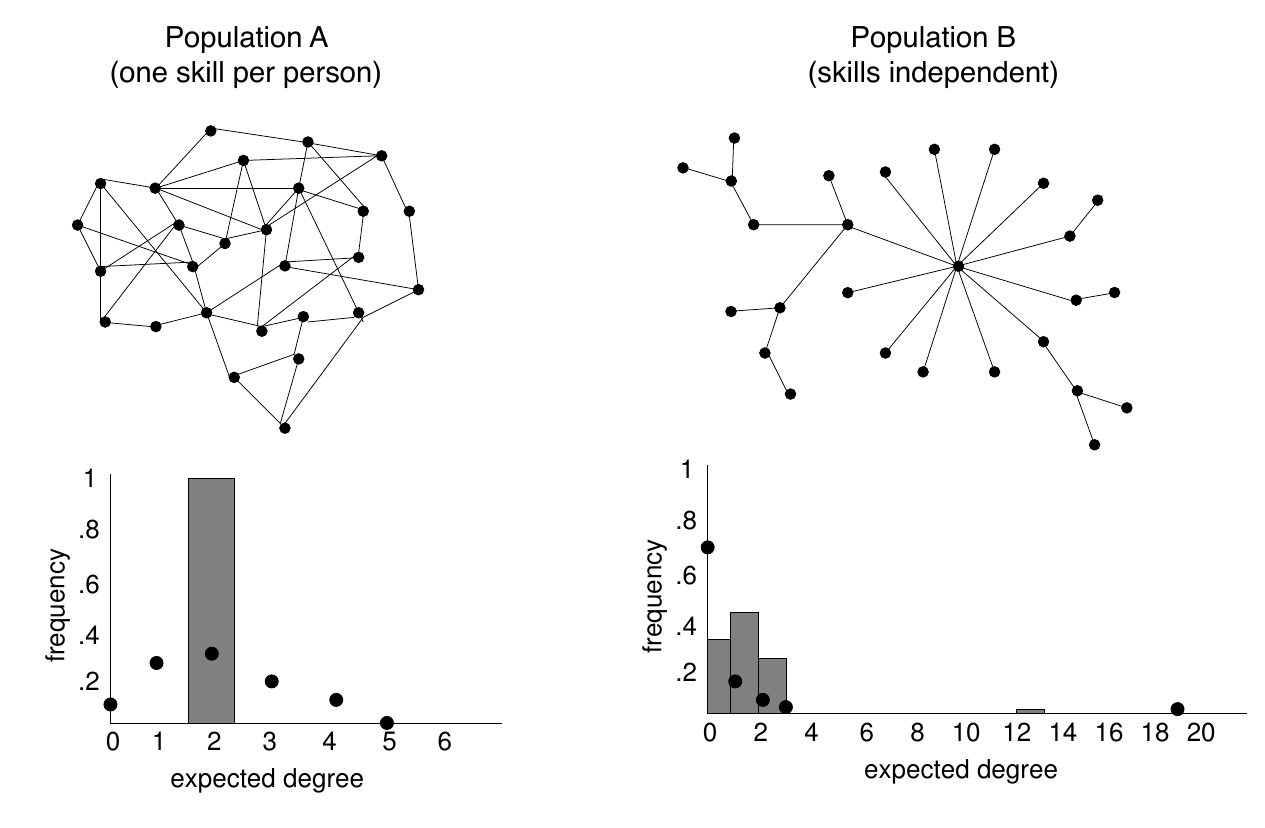}

\caption{\label{fig:Two populations Network Structure}A population with types
and a population with independent skills.}
\end{figure}

This pair of examples highlights the value of a more detailed treatment
of skills in modeling problem solving. Although it is possible to
model the heterogeneity of problem solvers more simply--for example,
via a one-dimensional ability level (\citealp{Gautier2002}, \citealp{Shi2002}),
by giving each individual a type or speciality (\citealp{Hamilton2000}),
or some combination of the two (\citealp{Jovanovic1994}), these modeling
choices are not necessarily benign. 

Allowing for correlations between skills can generate linking behavior
that non-monotonic in the number of skills an individual has--a case
that is both relevant empirically and not present in other models
of collaboration networks. For example, consider the population summarized
in the first panel of Figure \ref{fig:Two-populations}. 
\begin{figure}[h]
\includegraphics[width=1\textwidth]{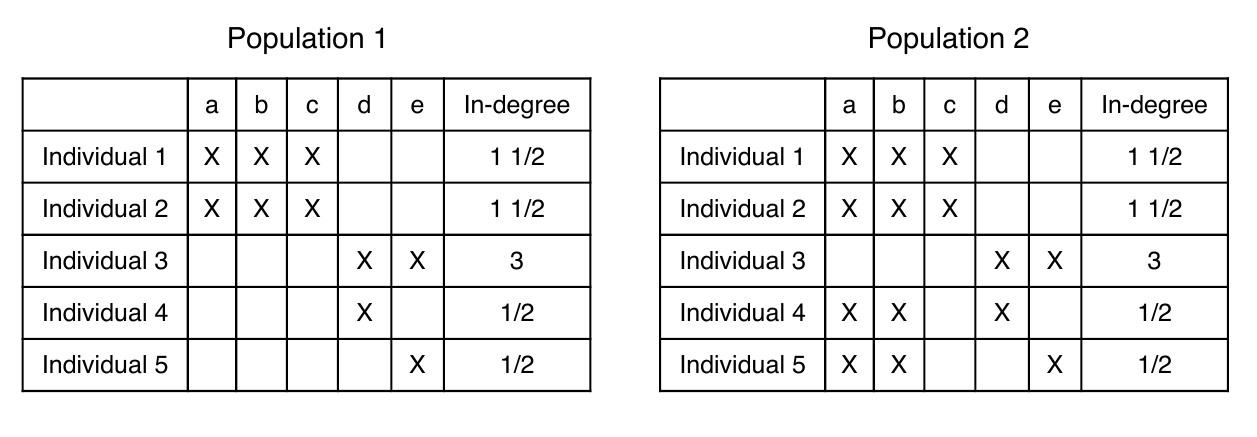}

\caption{\label{fig:Two-populations}Two populations of problem solvers with
five skills $S=\{a,b,c,d,e\}$--$n$ copies of each individual yields
a population of size $N=5n$.}
\end{figure}
 The problem requires five skills, $S=\left\{ a,b,c,d,e\right\} $,
which are distributed across a population of $N=5n$ as indicated.
Each skill is held by exactly $2n$ individuals, and thus no skill
is rarer than the others. Traditionally, we might condense the information
contained in this table into a single-dimensional measure of ability--individuals
1 and 2 have the most skills, and therefore, we might expect them
to have the most links. However, despite having fewer skills, individual
3 will, on average have more links. This is because while neither
of individual 3's skills are rare by themselves, in combination, they
are both rare, and useful to a large fraction of the population. This
means that agent 3 will have more links than a tally of her individual
skills might predict----in other words, the value of a combination\emph{
}of skills may be greater than the sum of its parts. Figure \ref{fig:Third-population}
shows another example. 
\begin{figure}[h]
\includegraphics[width=0.5\textwidth]{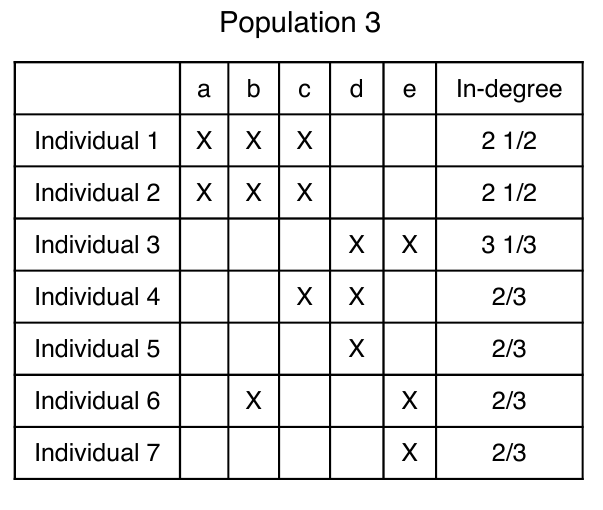}

\caption{\label{fig:Third-population}A third populations of problem solvers
with five skills $S=\{a,b,c,d,e\}$--$n$ copies of each individual
yields a population of size $N=7n$.}
\end{figure}
In this case, skills $d$ and $e$ are more common than skill $a$,
but the individual with skills $d$ and $e$ gets more links, showing
that an individual with common skills might even have more links than
an individual with rarer skills.

Examples of this phenomenon are not difficult to find. For example,
consider a tech entrepreneur who has experience in both computer programming
and marketing. Neither of those skills are rare individually, but
together they are quire rare. Individuals with this combination of
skills will be in high demand, and will have more collaborative connections
than those without that combination of skills. This further suggests
that in problem solving communities such as those in knowledge-based
industries, entrepreneurial firms, and academic research, models in
which individuals are scored on a one-dimensional ability scale will
fail to capture the full effect of variation between individuals. 

Pushing this point a bit further suggests another implication of Theorem
\ref{thm:SuperModularity}--because degree is a supermodular function
of a problem solver's complete set of skills, it is not generically
possible to assign prices to individual skills in a way that captures
her degree on the social network. This means that examining the supply
and demand of single skills in isolation does not necessarily capture
an individual's value to a community of problem solvers.
\begin{cor}
There need exist no vector of prices, $\mu$, such that $\sum_{a\in A}\mu_{a}=d(A)$.
\end{cor}
To further emphasize this point, consider the skill distribution shown
in the second panel of in Figure \ref{fig:Two-populations}. This
distribution is identical to that in the first panel, except that
individuals 4 and 5 have been given an extra skill. However, gaining
these skills does not influence the degree of either. In fact, endowing
them with those skills does not change any part of the degree distribution.
Skills $a$ and $b$ have value to individuals 1 and 2, but not to
individuals 3 or 4--clearly, no linear weighting of the individual
skills could produce such a pattern. Two characteristics of this model
contribute to this result. First, skills are bundled within a person.
Thus, in evaluating a collaborator, that person's combination of skills
must be considered as a unit. Second, the individuals in the model
have an incentive to minimize the number of collaborators they work
with. Together, these two factors mean that an individual's value
to the collaborative community may be more than the sum of her individual
skills.

These results have implications for empirical models of labor market
outcomes. If we assume that skills contribute to outcomes independently,
then we will underestimate the value of particular combinations of
skills and overestimate the value of other combinations. Moreover,
this non-linear relationship between skills and degree means that
if we observe skills imperfectly, the amount of variation in outcomes
that is explained by an observed set of skills will decline dramatically
in the set of skills we can observe. Additionally, the fact that an
individual's degree depends on both the distribution of skills in
the population and the set of skills she already has suggests that
in collaborative fields, optimal training decisions are highly individualized.\footnote{To see this, consider the Shapley value of an additional skill. When
an each individual requires only one collaborator to solve a problem,
the individual's degree is , (the results are similar for the case
where multiple collaborators may be required). Using degree, $d\left(A_{i}\right)=\sum_{C\subseteq A_{i}}\frac{\Psi\left(S\backslash C\right)}{\sum_{D\subseteq S\backslash C}\Psi\left(C\cup D\right)}$,
as a value function (note that $d\left(.\right)$ satisfies both $d\left(\emptyset\right)=0$
and superadditivity), the Shapely value of a skill, $a$, to individual
$i$: $\phi_{a,i}\left(d\right)=\sum_{B\subseteq A_{i}\backslash\left\{ a\right\} }\frac{1}{\left({\left|A_{i}\right|\atop \left|B\right|}\right)}\left(\sum_{C\subseteq B}\frac{\Psi\left(S\backslash\left(C\cup a\right)\right)}{\sum_{D\subseteq S\backslash\left(C\cup a\right)}\Psi\left(\left(C\cup a\right)\cup D\right)}\right)$,
which depends on both the existing skill population and the individual's
current set of skills. }

\section{\label{sec:Bernoulli Skills Model}The Distribution of Skills and
the Structure of the Collaboration Network: The Bernoulli Skills Model}

In this section, I look at how the distribution of skills in the population
affects the degree distribution of the collaboration network.

\subsection{The Bernoulli Skills Model}

Thus far, I have considered results for a general skill population.
However, when we turn to the effects of skill populations on network
structure, it is difficult to obtain clear predictions using such
a general construction. Therefore, in this section I will consider
a special case, where skills are distributed independently with equal
probability--that is, $Prob\left(a_{i}\in A|a_{j}\in A\right)=Prob\left(a_{i}\in A\right)=p\,\forall\,i\ne j\in S$.
I call this special case \emph{the Bernoulli Skills Model} because
each individual's skill set can be thought of as the result of a set
of $M$ Bernoulli trials, each with probability $p$ of success. This
means that the distribution of skill set sizes in the population is
binomial, implying that the fraction problem solvers who have a \emph{particular
}set of $k$ skills is $\Psi\left(A\right)=p^{k}\left(1-p\right)^{M-k}$,
and the fraction having \emph{any} $k$ skills is $\left({M\atop k}\right)p^{k}\left(1-p\right)^{M-k}$.\footnote{Note that as before, I will be using the \emph{frequency }distribution.
This means that implicitly I will be assuming that the population
is large enough for every skill combination in the distribution to
be represented at least once: $N=p^{-M}$. For small $M$, this assumption
poses no serious problem. However, as $M$ increases, population sizes
become quite large. Thus, the Bernoulli Skills model will provide
more insight in cases where there are fewer relevant skills. } This special case has several characteristics which make it a useful
starting point. First, because skills are uncorrelated and occur with
equal frequency, individuals with the same number of skills will,
in expectation, have the same degree. This means that the only variable
of interest is skill set size, $k$. Second, because this model has
only 2 parameters--$M$ and $p$--it is easy to see how changes in
the distribution of skills in the population affect the structure
of the collaboration network.

\subsection{Degree Distribution of the Bernoulli Skills Model}

Let $\Delta$ denote the distribution of expected degree. That is,
$\Delta\left(d\right)$ is the fraction of the population who have
expected degree $d$, where the expectation is taken over all $g\in\Gamma\left(\Psi\right)$.
In this particular case, I will use a convenient shorthand: $\Delta_{M,p}$
represents the distribution of expected degree when $M$ skills are
independently distributed with probability $p$. 

Theorem \ref{thm:d(A) for independent skills} states the closed form
expression for the degree of an individual in the Bernoulli Skills
Model. 
\begin{thm}
\label{thm:d(A) for independent skills}Suppose each individual in
a population faces a problem, $\omega\left(S\right)$, requiring $M$
skills. If the skills are distributed independently with $Prob\left(a\right)=p\,\forall\,a\in S$,
then the expected degree of an individual with $k$ skills is \textup{
\[
E\left[d(k)\right]=p^{M}\left[\left(\frac{1-p+p^{2}}{p^{2}}\right)^{k}-1\right]
\]
}\end{thm}
\begin{proof}
Since $\Sigma\left(\Psi\right)=2^{S}$, every individual needs to
make only one link. Thus, we can write $E\left[d\left(A\right)\right]=\sum_{C\subseteq A}\frac{\delta\left(C\right)}{\sigma\left(C\right)}$,
where $\delta\left(C\right)$ is the fraction of the population who
need skill set $C$ and $\sigma\left(C\right)$ is the fraction of
the population who can provide skill set $C$. Since the skills are
independent, we can separate this sum according to the size of the
skill set required. If we start with the individuals lacking exactly
one skill in $A_{i}$ and end with the individuals needing all of
the skills, we obtain the following sum: 
\begin{eqnarray*}
E\left[d\left(A\right)\right] & = & \sum_{i\in A}\frac{p^{M-1}\left(1-p\right)}{p}+\sum_{i,j\in A}\frac{p^{M-2}\left(1-p\right)}{p^{2}}+...+\frac{p^{M-k}\left(1-p\right)^{k}}{p^{k}}\\
 & = & \sum_{i=1}^{k}\left({k\atop i}\right)\frac{p^{M-i}\left(1-p\right)^{i}}{p^{i}}\\
 & = & p^{M}\left[\left(\frac{1-p+p^{2}}{p^{2}}\right)^{k}-1\right]
\end{eqnarray*}

\end{proof}
Note that when skills are independent, an individual's degree on the
network depends only on the size of her skill set, $k$. Therefore,
in the Bernoulli Skills Model, it is appropriate to interpret the
size of an individual's skill set as her ``ability''--something
that we cannot do in the more general case (recall Figure \ref{fig:Two-populations}
in the previous section). Moreover, in the Bernoulli Skills Model,
degree is a superadditive function of the number of skills a person
has: $E\left[d\left(k+1\right)\right]>E\left[d\left(k\right)\right]+E\left[d\text{\ensuremath{\left(1\right)}}\right]$.
This suggests a corollary to Theorem \ref{thm:d(A) for independent skills}.
\begin{cor}
Suppose each individual in a population faces a problem, $\omega\left(S\right)$,
requiring $M$ skills. If the skills are distributed independently
with $Prob\left(a\right)=p\,\forall\,a\in S$, then degree is a superadditive
function of the size of her skill set, $k=\left|A_{i}\right|$. Moreover,
her degree is strictly increasing in $k$. 
\end{cor}
However, we still cannot price the skills individually in such a way
that we capture degree, despite the fact that skills are independently
distributed.
\begin{thm}
\label{thm: Bernoulli non-linearity}Suppose each individual in a
population faces a problem, $\omega\left(S\right)$, requiring $M$
skills . If the skills are distributed independently with $Prob\left(a\right)=p\,\forall\,a\in S$,
then there exists no vector of prices, $\mu$, such that $\sum_{a\in A}\mu_{a}=E\left[d\left(A\right)\right]$
for all $A\subseteq S$.\end{thm}
\begin{proof}
Any such vector would be required to set $\mu_{a}=d\left(a\right)=p^{M-2}(1-p)$
for all $a\in S$. But that would imply that $d\left(A\right)=kp^{M-2}(1-p)$
for $\left|A\right|=k$. This is clearly not true for $k>1$.
\end{proof}
This superadditive relationship between skills and degree has a dramatic
affect on the distribution of links in the population. The distribution
of degree in the network is 
\[
\Delta_{M,p}\left(d\right)=\frac{1}{\left(1-p\right)^{M}}\left({M\atop k\left(M,p,d\right)}\right)\left(\frac{p}{1-p}\right)^{k\left(M,p,d\right)}
\]

where $k\left(M,p,d\right)=\frac{\ln\left(\frac{d}{p^{M}+1}\right)}{\ln\left(\frac{1-p+p^{2}}{p^{2}}\right)}$.
Note that although the distribution of ``ability'' (skill set size)
in the Bernoulli Skills Model is binomial, and thus symmetric, the
distribution of degree is highly skewed. Individuals with a few extra
skills will have many additional combinations of skills, so individuals
with marginally larger skill sets will have dramatically more links.
As a result, the distribution of degree has much higher variance than
the distribution of ability. Figure \ref{fig: Skills and Degree Dist}
shows an example of ability and degree distributions with $M=10$
and $p=\frac{1}{2}$.

\begin{figure}[h]
\includegraphics[width=1\textwidth]{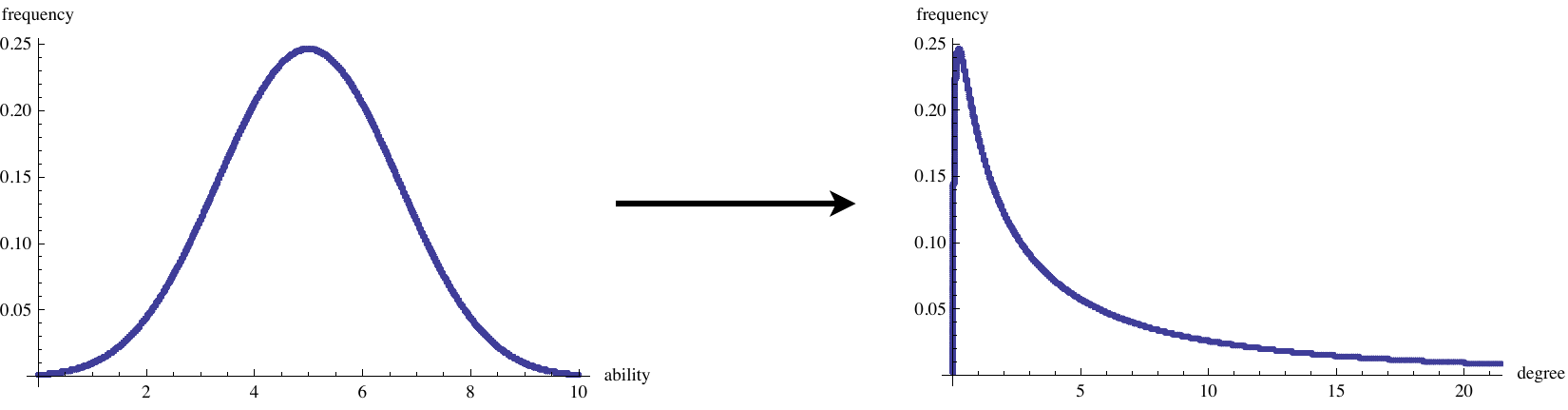}

\caption{\label{fig: Skills and Degree Dist}Distributions of ability and degree
in a Bernoulli Skills Network with $M=10$ and $p=\frac{1}{2}$. }
\end{figure}

The resulting network has the kind of long-tailed degree distribution
we expect to see in a collaboration network. Interestingly, this long-tailed
distribution emerges, despite the lack of a preferential attachment
dynamic. If such a dynamic were added on top of this model, the degree
distribution would become even more skewed.

This transformation of the ability distribution has an empirical implication--individuals
who have a disproportionate number of links in a collaboration network
may not have a disproportionate number of skills. In particular, individuals
who are only slightly more skillful may be dramatically more important
in the collaborative community. This means that collaborative communities
may have ``superstars''--individuals with a disproportionate influence
in the community--even when skills are distributed relatively evenly
in the population--while great innovators like Thomas Edison are undoubtedly
more skilled than the average entrepreneur, they are likely not 1000
times as skilled, as their levels of productive output might suggest.
Moreover, modeling individuals as having skill sets may account for
some unexplained variation in collaborative outcome\textbf{s.}

\subsection{Problem Difficulty and the Distribution of Degree: Comparative Statics
on the Bernoulli Skills Model}

While most collaboration networks share the skewed degree distribution,
there is some variation between them. In this section, I use comparative
statics of the Bernoulli skills model to examine how changes in the
distribution of skills change the amount of inequality in the distribution
of links. 

In the following, I will use the Gini coefficient as my measure of
distributional equality (See Appendix C for a discussion of the gini
coefficient in this case). Using the fact that $\Delta_{M,p}\left(d\right)=\left({M\atop k\left(M,p,d\right)}\right)p^{k\left(M,p,d\right)}\left(1-p\right)^{M-k\left(M,p,d\right)}$
for $k\left(M,p,d\right)=\frac{\ln\left(\frac{d}{p^{M}}-1\right)}{\ln\left(\frac{1-p+p^{2}}{p^{2}}\right)}$
the gini coefficient of the degree distribution $\Delta_{M,p}$ is
\[
G\left(M,p\right)=1-\frac{\left(1-p\right)^{2M}}{\left(1-p^{M}\right)}\left[\sum_{k=0}^{M}\left(\left({M\atop k}\right)\left(\frac{p}{1-p}\right)^{k}*\sum_{j=0}^{k}\left(\left({M\atop j}\right)\left(\frac{p}{1-p}\right)^{j}*d\left(M,p,k\right)\right)\right)\right]
\]
As is conventional, values of $G\left(p,M\right)$ closer to 1 indicate
a more unequal distribution of degree, while values closer to 0 indicate
a more equitable distribution of degree.

Recall that the Bernoulli Skills model has only two parameters: the
number of skills required to solve the problem ($M$) and the fraction
of those skills possessed by the average individual ($p$). Figure
\ref{fig:Example Networks M,p-1} shows the Gini coefficient for various
values of $M$ and Figure \ref{fig:Gini for p and M} shows it for
various values of $p$. The Gini coefficient is increasing in $M$,
meaning that as problems require more skills, the distribution of
degree becomes increasingly skewed towards a few high-degree nodes.
The Gini coefficient is decreasing in $p$, meaning that as the probability
of having a given skill goes down, the network becomes more skewed
(see Figure \ref{fig:Gini for p and M}). These results are summarized
in Theorem \ref{thm:Comparative Static M p}.
\begin{thm}
\label{thm:Comparative Static M p}Suppose each individual in a population
faces a problem, $\omega\left(S\right)$, requiring $M$ skills and
further suppose the skills are distributed independently with $Prob\left(a\right)=p\,\forall\,a\in S$.
Let $G\left(p,M\right)$ be the Gini Coefficient for the resulting
degree distribution, $\Delta_{M,p}$. Then
\begin{enumerate}
\item $G\left(p,M\right)$ is strictly increasing in $M$. That is, the
distribution of links in the collaboration network is more uneven
when the problem being solved requires more skills. 
\item As $G\left(p,M\right)$ is strictly decreasing in $p$. That is, the
distribution of links in the collaboration network is more uneven
as the probability of having each skill gets lower.
\end{enumerate}
\end{thm}
\begin{figure}[h]
\includegraphics[width=0.9\columnwidth]{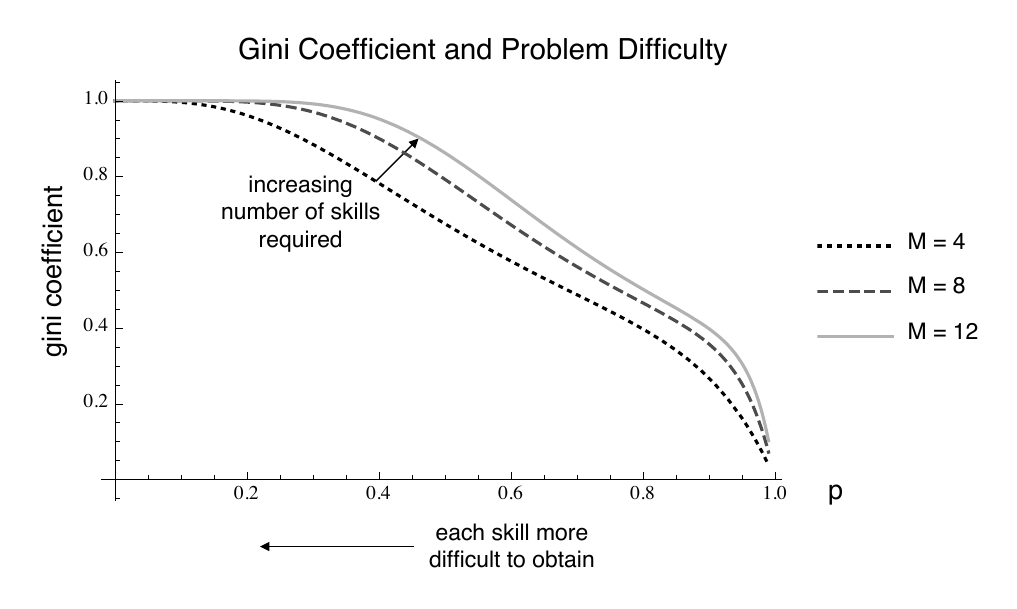}

\caption{\label{fig:Gini for p and M}The gini coefficient, $G\left(p,M\right)$,
for different values of $M$ and $p$.}
\end{figure}

Together, these results can be interpreted in terms of the difficulty
of the problem being faced. A problem is more difficult if it requires
many skills ($M\uparrow$), or if the average individual has only
a few of them ($p\downarrow$). This means that the above can be recast
as a comparative static on problem difficulty. As problems become
increasingly difficult ($M\uparrow$ and $p\downarrow$), the degree
distribution of the network becomes more skewed, and a small number
of ``superstars'' will dominate the collaborative community. Figures
\ref{fig:Example Networks M,p-1} and \ref{fig:Example Networks M,p}
illustrate the changes in the network as these parameters shift.

\begin{figure}[h]
\includegraphics[width=0.9\textwidth]{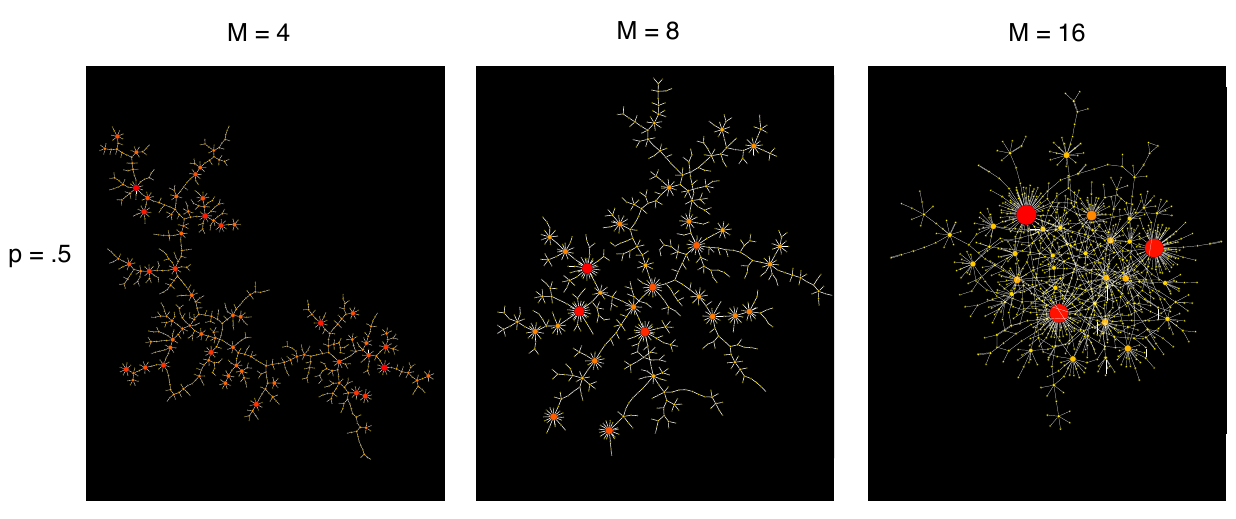}

\caption{\label{fig:Example Networks M,p-1}Examples of Bernoulli Skills Networks
with $N=1000$}
\end{figure}

\begin{figure}[h]
\includegraphics[width=0.9\textwidth]{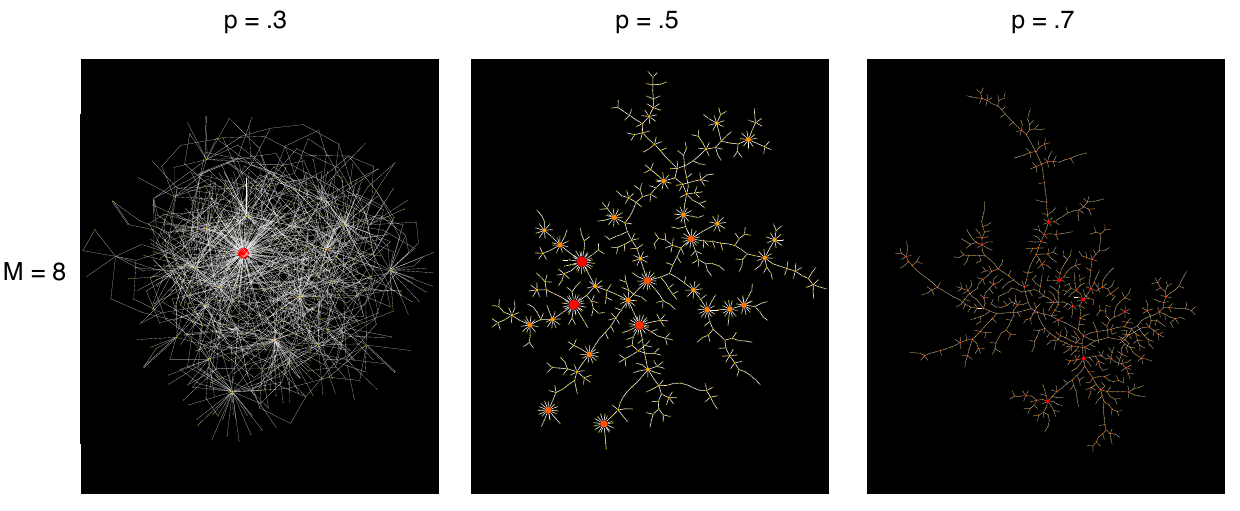}

\caption{\label{fig:Example Networks M,p}Examples of Bernoulli Skills Networks
with $N=1000$}
\end{figure}

\section{\label{sec:Skill-Ladders}Skill Ladders: Skill Hierarchy and Degree}

In the previous section, I considered a special case in which skills
are entirely uncorrelated. In this section, I consider the effect
of correlations between skills. I divide the skills into multiple
skill ladders, where the skills within a ladder build on one another.
I show that when skills are correlated in this way, the degree distribution
of the collaboration network becomes even more unequal than when skills
are uncorrelated. This suggests that when skills in a community are
structured in this way, a very small number of individuals will tend
to dominate the collaboration network.

\subsection{The Ladder Model}

In order to consider the effect of correlations between skills, in
this section, I will need a second special case, which I will call
the Ladder Model. I will define a \emph{ladder }to be an ordered set
of skills, $L=\left\{ a_{1},a_{2},a_{3}...a_{l}\right\} \subseteq S$,
such that $Prob\left(\mbox{have }a_{i}|\mbox{have }a_{i+1}\right)=1$
(that is, such that any individual who has the $i^{th}$ skill in
the set must have all of the skills that precede it in the set).\footnote{\citealp{Page2007} introduces this concept of skill ladders.}
Here, I consider a special case where the skills in $S$ are partitioned
into $m$ ladders of equal length.\footnote{Obviously, since the length of a ladder is an integer, there will
only be equal-length ladders if $m$ divides $M$ evenly. To simplify
the exposition, I have written the following as if this is true. However,
all of the the following results hold if the ladders are equal length
up to integer constraints, which allows for cases where $m$ does
not divide $M$ evenly.} The set of ladders is denoted $\hat{S}=\left\{ L_{1}...L_{m}\right\} $.
Figure \ref{fig:FourLadders} shows an example with 12 skills arranged
into four ladders.

\begin{figure}[H]
\includegraphics[clip,scale=0.5]{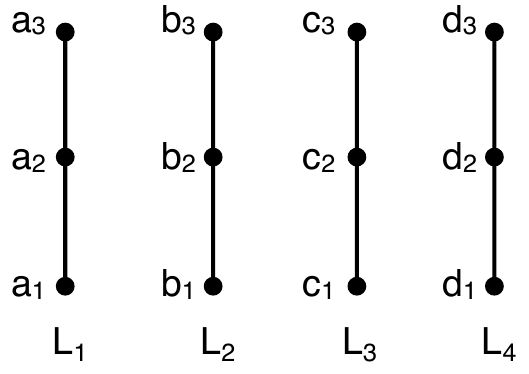}

\caption{\label{fig:FourLadders}12 skills arranged into four ladders of equal
length.}
\end{figure}

I will call an individual who has all of the skills in a single ladder
an ``expert'' in that ladder, and I will call the set of ladders
that individual $i$ is an expert in $\hat{A}_{i}\subseteq\hat{S}$. 

One additional assumption will allow me to compare the results in
this section to the results of the Bernoulli skills model. I will
assume that the conditional probability of having the next skill in
a ladder is the same for all skills--that is, $Prob\left(\mbox{have }a_{i}|\mbox{ have }a_{i-1}\right)=p$
for all $a_{i}$.\footnote{In other words, I assume that putting a skill at the end of a ladder
doesn't change the essential difficulty of obtaining that skill. If
one instead makes a skill at the top easier to obtain (eg: because
it builds on previous experience) or harder to obtain (eg: because
they are more demanding than the skills that came before), then the
results remain qualitatively unchanged.} The probability of being an expert in a ladder of $l$ skills is
then $p^{l}$. Thanks to this assumption, the case where $m=M$ corresponds
to the Bernoulli skills model. On the other extreme, $m=1$, and all
of the skills are arranged in a single ladder. Before considering
ladders of arbitrary length, I will first look at this case, where
$m=1$.

\subsection{\label{exa:A-single-ladder}Example: a single ladder of skills}

Suppose the skills in the set $S$ comprise a single ladder of length
$M$. Because $Prob\left(\mbox{have }a_{i}|\mbox{have }a_{i+1}\right)=1$
$\forall i\in S$, an individual's skill set can be represented by
the number of skills she has ($\left|A_{i}\right|=k$ implies $A_{i}=\left\{ a_{1},a_{2}...a_{k}\right\} $)
. Because individuals can be ranked in order of the number of skills
they have, this case corresponds to a model in which each individual
has a one-dimensional ability measure.

The linking behavior in this case is very simple. The only individuals
who have skill $a_{M}$ are those who also have skills $a_{1}...a_{M-1}$.
Anyone who doesn't have all $M$ skills links to someone who does.
The resulting collaboration network is a set of isolated stars, each
with $\frac{1-p^{M}}{p^{M}}$ links, on average. Figure \ref{fig:Compare Network with Ladder}
compares the network structure in the case with one skill ladder $\left(m=1\right)$
to the network structure in the Bernoulli skills model, where skills
are independent $\left(m=M\right)$. The populations that make up
these two networks are similar--they have the same number of problem
solvers, the same number of skills, and the same probability of having
an additional skill. This means that the probability of an individual
having all of the skills required to solve the problem is the same
in both networks. Moreover, in both, exactly one individual has all
of the skills required. However, the two networks have a much different
structure. When all of the skills build on each other, the network
is centered around a single, high-degree node. When skills are independent,
the network structure is much more distributed.\footnote{Note that this could be interpreted as a prediction about organizational
structures in different industries. In industries where value is created
through the exploitation of existing knowledge (eg: manufacturing),
skills tend to build on one another, and it is viable to rank workers
by ability. However, as demand for workers shifts towards industries
that create value by generating new knowledge, we expect skills to
be arranged in ladders less often. This model predicts a corresponding
shift from hierarchical organizational structures towards more distributed
organizational structures. Empirical evidence seems to support that
prediction. Traditional organizational structures were hierarchical
(the theoretical underpinnings are explored in \citealp{Rosen1982}).
However, evidence indicates that organizational structure within firms
is changing--hierarchical structures are flattening, and workplaces
are becoming more decentralized (see, for example, \citealp{Bresnahan2002}
and \citealp{Rajan2006}). }

\begin{figure}[h]
\includegraphics[width=0.8\columnwidth]{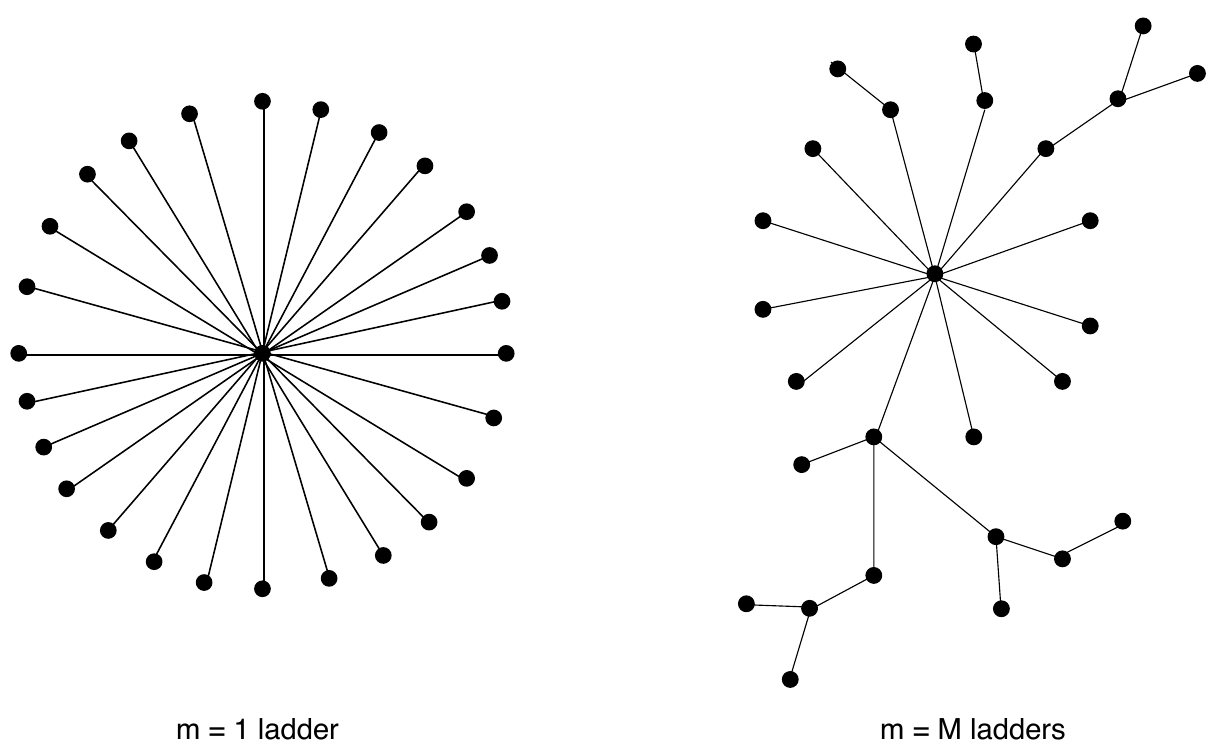}

\caption{\label{fig:Compare Network with Ladder}Ladder Skills Networks with
$N=27$ and $M=3$. }
\end{figure}

\subsection{Results for $m$ ladders}

Now consider the more general result. Suppose the skills are arranged
in $m$ equal-length ladders. As in the previous example, only experts
obtain links. Thus, a model with $m$ ladders reduces to a Bernoulli
skills model with $m$ independent skills. Theorem \ref{thm:degree with m ladders}
presents a closed-form expression for a individual problem solver's
degree in the case with $m$ skill ladders.
\begin{thm}
\label{thm:degree with m ladders}If $\Psi$ is a distribution of
skills such that $\hat{S}=\left\{ L_{1}...L_{m}\right\} $ is a partition
of $S$ into $m$ equal-length ladders with $Prob\left(\mbox{have }a_{i}^{j}|\mbox{ have }a_{\left(i-1\right)}^{j}\right)=p$,
then an individual with the skill set $A$ will have expected degree
$E\left[d(A)\right]=p^{\frac{M}{m}}\left[\left(\frac{1-p^{\frac{M}{m}}+p^{2\frac{M}{m}}}{p^{2\frac{M}{m}}}\right)^{k}-1\right]$
, where $k$ is the number of disciplines the individual is an expert
in.\end{thm}
\begin{proof}
The ladders are of equal length, so the length of a single ladder
is $\frac{M}{m}$, and the probability that an individual is an expert
in any one ladder is $p^{\frac{M}{m}}$. An individual receives a
link only if she is an expert in a field. Define a new set of skills
that correspond to the set of ladders: $\hat{S}=\left\{ L_{1}...L_{m}\right\} $.
The individual's new skill set is $\hat{A}_{i}$, where $L_{k}\in\hat{A}_{i}$
if she is an expert in ladder $L_{k}$. Each of these new skills has
a probability equal to the probability of being an expert in that
field, so define $\hat{p}=p^{\frac{M}{m}}$. The probability of being
an expert in a particular ladder is independent of the probability
of being an expert in any other ladder, so this problem reduces to
one with $m$ independent skills, with probability $p^{\frac{M}{m}}$.
The result then is a simple extension of Theorem \ref{thm:d(A) for independent skills}.
\end{proof}
We can now do a comparative static on the number of skill ladders,
to see how the amount of hierarchy in the skill pool affects the structure
of the collaboration network. Theorem \ref{thm:Number of Ladders}
indicates that in communities where skills are more hierarchical,
the degree distribution of the collaboration network is more skewed,
and a few experts become extremely high-degree hubs in the network.
Figure \ref{fig:Ladders Gini} shows how the gini coefficient depends
on the number of ladders for the case where $M=6$ (note that the
bottom curve $\left(m=6\right)$ is equivalent to the Bernoulli Skills
Model with $M=6$ in Figure \ref{fig:Gini for p and M}).

\begin{figure}[h]
\includegraphics[width=0.9\columnwidth]{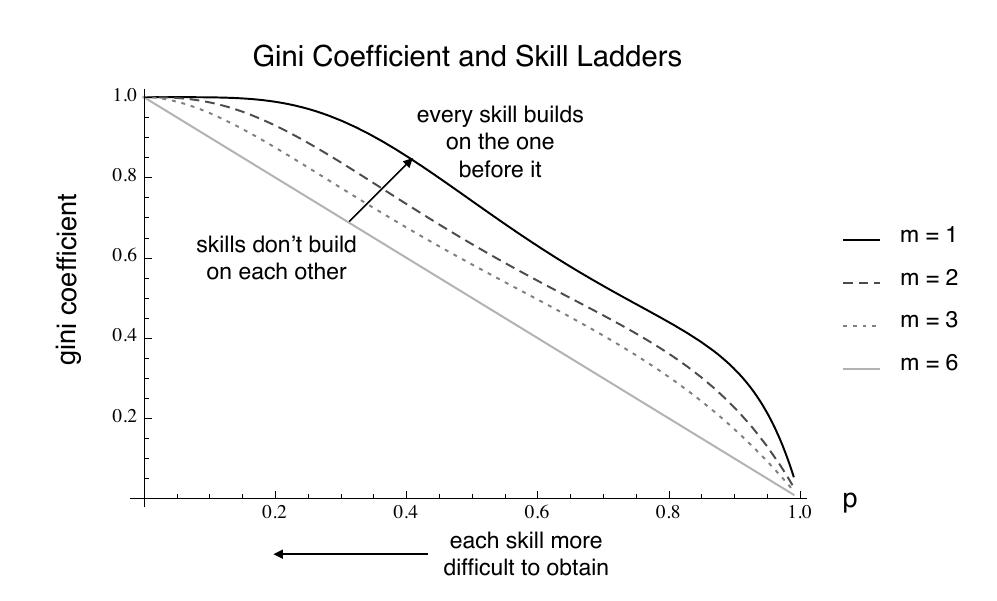}

\caption{\label{fig:Ladders Gini}The gini coefficient, $G\left(p,M,m\right)$,
for different values of $m$ and $p$ and $M=6$. }
\end{figure}

\begin{thm}
\label{thm:Number of Ladders}Suppose $S$ skills are arranged in
$m$ ladders of equal length, with constant conditional probability
$Prob\left(\mbox{have }a_{i}^{j}|\mbox{ have }a_{i-1}^{j}\right)=p$
and $Prob\left(\mbox{have }a_{1}^{j}\right)=p$ $\forall j=1...m$.
The gini coefficient of the resulting network is decreasing in the
number of ladders, $m$. That is, when there are fewer skill ladders,
the degree distribution becomes increasingly uneven.\end{thm}
\begin{proof}
A model with $m$ skill ladders is equivalent to a Bernoulli Skills
model with $m$ skills and $p=p^{M/m}$. The gini coefficient is thus
$G\left(p,\frac{M}{m}\right)$. Since this function is increasing
in the second argument (see Theorem \ref{thm:Comparative Static M p}),
it is decreasing in $m$.
\end{proof}

\section{Discussion: Labor Market Implications}

The more general treatment of individual heterogeneity that I present
in this paper has some interesting implications for labor markets
knowledge-based fields, where workers are richly heterogeneous and
collaborate to solve problems.\footnote{If we were to recast the relationship between collaborators as a relationship
between employees and firms, then an extension of this model would
also provide insights into the affect of heterogeneity on labor in
a non-collaborative context.} In particular, given that there is a relationship between the number
of collaborators an individual has and her output,\footnote{In the Bernoulli Skills model, where each individual has a complement,
the relationship is one-to-one. } the model presented in this paper can provide insight into individual
heterogeneity and output. The results of Section \ref{sec: Value of Model with Multiple Skills}
strongly indicate that the decisions we make when modeling heterogeneity
matter. A full exploration of the effects of complex heterogeneity
on labor markets is outside the scope of this work, but I will provide
a brief discussion here. 

Empirically, output in creative and knowledge-based fields is highly
concentrated among a small number of people.\footnote{\citealp{Newman2001}, \citealp{Moody2004}, \citealp{Acedo2006b},
and \citealp{Goyal2006} show that a small number of academics produce
the majority of papers written, \citealp{Rosen1981} observes a similar
pattern in in music, film, and textbook writing, and \citealp{Uzzi2005}
notes a similarly skewed distribution of output among directors, producers,
and other creative artists on Broadway.} This long-tailed distribution of productive output has implications
for the distribution of wages and welfare, and thus there has been
considerable interest in understanding why such a concentration in
labor demand occurs. \citealp{Rosen1981}'s model of production can
induce a long-tailed distribution when there is a high premium on
quality, and production technology decouples effort from output quantity
(eg: in creative industries, where a single performance or album can
be enjoyed by many consumers). However, such technologies are not
relevant in knowledge-based industries, where effort is not decoupled
from output volume. The long-tailed distribution of collaborative
interactions in this paper suggests an alternative explanation for
the observed variation in output. Moreover, the prediction that output
will become more skewed as problems become more difficult suggests
an empirical test in the context of aggregate labor demand. 

The non-linear relationship between skill sets and output also has
implications for the explanation of variation in labor outcomes when
individual heterogeneity is imperfectly observed. If we assume that
skills make a linear contribution to output, then we would expect
a linear decline in the amount of variation explained. However, if
the relationship between individual skills and output is non-linear,
as is suggested by the model in this paper, then the amount of variation
in output explained by observables will drop nonlinearly as well.
This suggests that decreases in the amount of variation in labor outcomes
explained by observed agent heterogeneity could be the result of a
shift from manufacturing--where skill could reasonably be considered
to be a one-dimensional variable--to knowledge-based industries, where
individual heterogeneity is more complex. A more thorough exploration
of the implications of skill-based labor heterogeneity on labor markets
would make an interesting extension.

\section{Conclusion}

In this paper, I have presented a model of collaboration network formation
in which individuals have heterogeneous skill sets and collaborate
to solve problems that none of them could solve as individuals. The
result is a collaboration network, the structure of which depends
on the distribution of skills in the underlying community. This framework
is extremely flexible and has the potential to provide insights into
a wide variety of questions that could not previously be answered.
Here, I have used an independent distribution of skills to look at
how the difficulty of problem-solving tasks might affect the structure
of the collaboration network. But the same general framework can be
used to explore a variety of questions about the collaborative process.
For example, one could use correlations between related skills to
explore how skill specialization affects the structure of a collaborative
community. A more elaborate set of correlations between skills might
provide insight into why some skills are under-provided. 

One of the strengths of the framework introduced here is that it lends
itself to extension and modification. One could use changes to the
payoff structure to examine how community structure depends on team
dynamics. For example, how does the structure of the community change
if teams benefit from skill overlap as well as skill complementarity?
And how does that change the role of individuals with bridging skills?
Time is another interesting dimension for future work. A dynamic model
would allow for a more complete model of search. Individuals might
find new collaborators by searching the networks of their current
collaborators. Such a model might also incorporate a geographic component,
making it more difficult for geographically dispersed individuals
to collaborate. In a longer-run model there is room to explore the
role of learning. This model assumes a fixed skill population and
endowments. However, individuals clearly develop new skills over time.
A model of the skill acquisition decision would pave the way for a
more dynamic model, in which the network and population both evolve
over time. In that case, there may be characteristics of the distribution
of problems that influence the long-run distribution of skills and
network structure. In the even longer term, it would be interesting
to see how shocks to the distribution of problems (eg: a dramatic
shift in the outlook of a field) would affect the skill population
and collaboration network. In particular, what factors affect the
robustness of a problem-solving population to changes in the distribution
of problems faced? Finally, the results of this paper are evidence
that more complex models of individual heterogeneity are both tractable
and crucial to understanding the structure and function of collaborative
communities.

\section*{Appendices: The Formation of Collaboration Networks among Individuals
with Heterogeneous Skills}

\subsection*{\label{sec:Appendix: Proof of Stabilty/Efficiency}Appendix A: Pairwise
Stability and Efficiency}

Briefly, a network is pairwise stable if no individual would prefer
to terminate an existing link, and if no pair of individuals would
prefer to add a link. Although this definition is usually used in
undirected networks, it works equally well in the current context.
Formally, a directed collaboration network, $g$, is \emph{pairwise
stable} if 
\begin{enumerate}
\item for all $ij\in g$, $u_{i}\left(g\right)\ge u_{i}\left(g-ij\right)$
and $u_{j}\left(g\right)\ge u_{j}\left(g-ij\right)$
\item for all $ij\notin g$, if $u_{j}\left(g+ij\right)>u_{j}\left(g\right)$
then $u_{i}\left(g+ij\right)<u_{i}\left(g\right)$
\end{enumerate}
Together, these two conditions ensure that links are mutual. That
is, if a network is pairwise stable, then both ends agree to maintain
the link. 
\begin{thm*}
Any complementarity network, $g\in\Gamma\left(\Psi\right)$, is pairwise
stable. In other words, $\forall ij\in g$ $u_{i}\left(g\right)\ge u_{i}\left(g-ij\right)$
and $u_{j}\left(g\right)\ge u_{j}\left(g-ij\right)$ and for all $ij\notin g$,
if $u_{j}\left(g+ij\right)>u_{j}\left(g\right)$ then $u_{i}\left(g+ij\right)<u_{i}\left(g\right)$.\end{thm*}
\begin{proof}
First, consider whether any individual wishes to unilaterally remove
a link, $ij\in g$. Removing this incoming link costs individual $j$
her share of the payoff from solving $i's$ problem $\left(\frac{1}{\left|C_{i}\right|+1}\ge0\right)$,
and thus she will never choose to do so. Individual $i$ chose a minimal
set of collaborators that allowed her to solve her problem, so removing
this outgoing link means that she can no longer solve the problem.
This would cost her the payoff $\left(\frac{1}{\left|C_{i}\right|+1}\ge0\right)$
from solving the problem, and thus she will also never choose to do
so. Finally, note that no individual would ever want to add an outgoing
link to a cost-minimizing collaboration network, because having chosen
a minimal set of collaborators, any additional link would require
her to further split her prize. \end{proof}
\begin{thm*}
Any cost minimizing collaboration network, $g\in\Gamma\left(\Psi\right)$,
is strongly efficient. In other words, $\sum_{i}u_{i}\left(g\right)\ge\sum_{i}u_{i}\left(g'\right)\,\forall\,g'\in G$.\end{thm*}
\begin{proof}
Because all value is generated from solving problems, the maximum
possible value in the network is $N$. Since solving problems is incentive
compatible and there is no loss, the problem solvers always extract
the maximum value from the network.
\end{proof}

\subsection*{\label{sec:Appedix: Proof of SuperModularity}Appendix B: General
Proof of Theorem 2}

An individual with the set $A\cup B$ will be able to help anyone
needing any subset of those skills. Let $\delta\left(C\right)$ be
the demand for a particular set of skills, $C$. In the general case,
$ $$\delta\left(C\right)=\Psi\left(S\backslash C\right)+\sum_{D\,:\,\Psi\left(C\cup D\right)=0}\Psi\left(S\backslash\left(C\cup D\right)\right)$.
The fraction who can supply the set $C$ is $\sigma\left(C\right)=\sum_{D\subseteq S\backslash C}\Psi\left(C\cup D\right)$.
Note that $\delta\left(C\right)$ and $\sigma\left(C\right)$ depend
only on the particulars of the problem $\left(S\right)$, the distribution
of skills $\left(\Psi\right)$, and the subset of skills $\left(C\right)$.
Thus, any individual with the skill set $A\cup B$ has expected degree
\[
E\left[d\left(A\cup B\right)\right]=\sum_{C\subseteq A\cup B}\frac{\delta\left(C\right)}{\sigma\left(C\right)}
\]

We can divide the problems that an individual with $A\cup B$ can
solve into three groups:
\begin{enumerate}
\item Requires only skills from set $A$: $C\subseteq A$
\item Requires only skills from set $B$, including at least one found only
in\\
 $B$ : $\left\{ C\,|\,C\subseteq B\,and\,\exists\,b\in C\,st\,b\in B\backslash A\right\} $
\item Requires at least one skill from each set that can only be found in
that set:\\
 $\left\{ C\,|\,C\subseteq A\cup B,\,where\,\exists\,a,b\in C\,s.t.\,a\in A\backslash B\,and\,b\in B\backslash A\right\} $
\end{enumerate}
Using this partition, we can write

\begin{eqnarray*}
E\left[d\left(A\cup B\right)\right] & = & \sum_{C\subseteq A}\frac{\delta\left(C\right)}{\sigma\left(C\right)}+\sum_{C\subseteq B\,and\,C\cap B\ne\emptyset}\frac{\delta\left(C\right)}{\sigma\left(C\right)}+\sum_{C\subseteq A\cup B\,and\,C\cap A,\,C\cap B\ne\emptyset}\frac{\delta\left(C\right)}{\sigma\left(C\right)}\\
 & = & E\left[d\left(A\right)\right]+\sum_{C\subseteq B\,and\,C\cap B\ne\emptyset}\frac{\delta\left(C\right)}{\sigma\left(C\right)}+\phi
\end{eqnarray*}
which implies that 
\begin{eqnarray*}
E\left[d\left(A\cup B\right)\right]+E\left[d\left(A\cap B\right)\right] & = & E\left[d\left(A\right)\right]+\sum_{C\subseteq B\,and\,C\cap B\ne\emptyset}\frac{\delta\left(C\right)}{\sigma\left(C\right)}+\phi+E\left[d\left(A\cap B)\right)\right]\\
 & = & E\left[d\left(A\right)\right]+\left(\sum_{C\subseteq B\,and\,C\cap B\ne\emptyset}\frac{\delta\left(C\right)}{\sigma\left(C\right)}+\sum_{C\subseteq A\cap B}\frac{\delta\left(C\right)}{\sigma\left(C\right)}\right)+\phi\\
 & = & E\left[d\left(A\right)\right]+E\left[d\left(B\right)\right]+\phi\\
 & \ge & E\left[d\left(A\right)\right]+E\left[d\left(B\right)\right]
\end{eqnarray*}

\subsection*{\label{sub:Gini Coefficient}Appendix C: Discussion of the Gini Coefficient
in the discrete case}

The Gini coefficient measures the area between the Lorenz curve of
a distribution (in this case, the distribution of expected degree),
and the line of equality. In the case of a discrete distribution with
values $y_{0}...y_{N}$ where $y_{i}<y_{i+1}$, the Lorenz curve is
a piecewise function connecting points $\left(F_{i},D_{i}\right)$
where $F_{i}=\sum_{k=0}^{i}\Delta\left(y_{k}\right)$ is the fraction
of individuals with strictly less than $y_{i}$ links, and $D_{i}=\frac{\sum_{k=0}^{i}\Delta\left(y_{k}\right)y_{k}}{\sum_{k=0}^{N}\Delta\left(y_{k}\right)y_{k}}$
is the fraction of the total number of links they hold. See Figure
\ref{fig:Lorenz Curve Example} for an example. The gini coefficient
for a discrete distribution is given by $G=1-\sum_{i=1}^{N}D_{i}\left(F_{i}-F_{i-1}\right)$.
Lower values of the gini coefficient indicate a more equal distribution
of links across individuals in the population, and higher values indicate
a more skewed distribution of links. The coefficient is which is $0$
when the distribution is perfectly equal (ie: the bottom $x\%$ of
the population holds exactly $x\%$ of the links) and $1$ when all
of the links are held by a single individual.
\begin{figure}[h]
\includegraphics[scale=0.5]{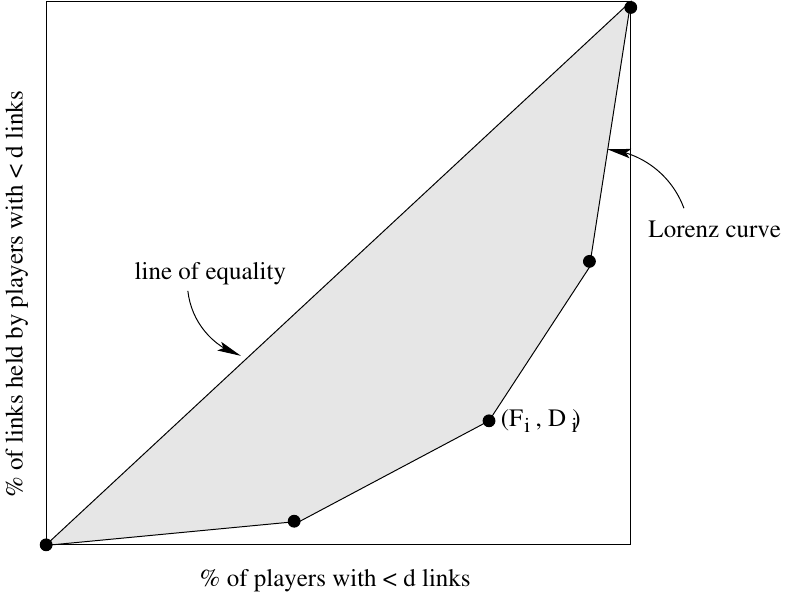}

\caption{\label{fig:Lorenz Curve Example}An example of the Gini coefficient
for a discrete distribution, $\Delta\left(y\right)$. In this case,
the random variable $y$ takes on one of five values, $y_{0}...y_{4}$.
The Gini coefficient is the area of the shaded region between the
line of equality and the Lorenz curve.}
\end{figure}

\bibliographystyle{nws}
\bibliography{SCN_paper}

\end{document}